\documentclass[%
 reprint,
 amsmath,amssymb,
 aps,
 pra,
floatfix,
]{revtex4-2}

\usepackage{graphicx}
\usepackage{dcolumn}
\usepackage{bm}
\usepackage{amsthm} 
\usepackage{amsmath, amssymb}
\usepackage[braket, qm]{qcircuit}
\usepackage{mathptmx}
\usepackage{braket}
\usepackage{amsfonts}
\usepackage{csquotes}
\usepackage{hhline}
\usepackage{relsize}
\usepackage{listings}
\usepackage{color}
\usepackage{physics}
\usepackage{soul}
\usepackage[linesnumbered,ruled,vlined]{algorithm2e}
\usepackage[hidelinks]{hyperref}
\usepackage{braket}

\makeatletter
\def\pdfstartlink@attr{attr{/Border[0 0 0]}}
\makeatother

\newtheorem{theorem}{Theorem}

\newtheorem{lemma}{Lemma}
\newtheorem{coro}{Corollary}

\begin{document}


\title{Local and Multi-Scale Strategies to Mitigate Exponential Concentration in Quantum Kernels}

\author{Claudia Zendejas-Morales}
\affiliation{Department of Mathematical Sciences, University of Copenhagen, Denmark}

\author{Debashis Saikia}
\affiliation{Department of Physics, Indian Institute of Science Education and Research, Thiruvananthapuram}


\author{Utkarsh Singh}
\affiliation{Department of Physics, University of Ottawa, 25 Templeton Street, Ottawa, Ontario, K1N 6N5 Canada}
\affiliation{National Research Council of Canada, 100 Sussex Drive, Ottawa, Ontario K1N 5A2, Canada}

\begin{abstract}
Fidelity-based quantum kernels can exhibit exponential concentration as the feature dimension or circuit expressivity grows, causing off-diagonal similarities to collapse and the Gram matrix to approach the identity. We study Qiskit implementations of local kernels that aggregate subsystem similarities and multi-scale kernels that combine several patch sizes. A Haar-random reference calculation shows that the mean raw similarity scale changes from $2^{-d}$ globally to $2^{-k}$ on $k$-qubit patches. Across six tabular datasets and $d\in\{4,6,\dots,20\}$, both constructions consistently reduce off-diagonal concentration relative to the global fidelity baseline. Effective-rank behavior is distinct from concentration, while SVM accuracy remains dataset-dependent.
\end{abstract}
\maketitle


\section{Introduction}\label{intro}

Kernel methods are a standard approach to nonlinear learning that separates representation from optimization. Given training data $\{\mathbf{x}_i\}_{i=1}^n \subset \mathbb{R}^d$, a positive semidefinite kernel $k: \mathbb{R}^d \times \mathbb{R}^d \to \mathbb{R}$ defines a Gram matrix $K\in\mathbb{R}^{n\times n}$, with entries
\begin{equation}
K_{ij} = k(\mathbf{x}_i,\mathbf{x}_j).
\label{eq:gram_def}
\end{equation}

Classical algorithms such as support vector machines can then be trained using $K$ \cite{scholkopf2002learning}.

Quantum kernel methods instantiate $k$ using a quantum feature map. A parameterized circuit $U_{\phi}(\mathbf{x})$ prepares an encoded state
\begin{equation}
\ket{\psi(\mathbf{x})} = U_{\phi}(\mathbf{x})\ket{0}^{\otimes d}.
\label{eq:feature_state}
\end{equation}
Here, $\phi$ denotes the feature map implemented by the circuit family $U_{\phi}$, and $\ket{\psi(\mathbf{x})}$ denotes the corresponding encoded quantum state. The choice of data encoding (feature map) critically shapes the induced kernel and, consequently, the expressivity and generalization behavior of the resulting quantum model \cite{schuld2021effect}. A common choice of kernel is the squared fidelity
\begin{equation}
k_{\mathrm{fid}}(\mathbf{x},\mathbf{x}') = \left|\langle\psi(\mathbf{x})\mid\psi(\mathbf{x}')\rangle\right|^2.
\label{eq:fidelity_kernel}
\end{equation}

This approach is attractive in the noisy intermediate-scale quantum regime because the optimization problem remains classical while the quantum device provides a structured, potentially high-dimensional embedding \cite{schuld2019quantum,havlivcek2019supervised,preskill2018quantum, biamonte2017quantum}.

A central obstacle for fidelity-based quantum kernels is \emph{exponential concentration}. Empirically and theoretically, as the number of qubits and$/$or circuit expressivity increase, overlaps between distinct encoded states can concentrate around a data-independent value \cite{thanasilp2024exponential,Agliardi2025,singh2025}. In the regime where feature-map states resemble random states, off-diagonal entries of $K$ become small and the Gram matrix approaches a low-variance form. This collapse can suppress label-relevant structure and can make reliable estimation of $k_{\mathrm{fid}}(\mathbf{x},\mathbf{x}')$ shot-expensive \cite{thanasilp2024exponential}. For kernel methods, concentration is therefore both a statistical issue (loss of informative variation across $K$) and an operational issue (measurement cost).

Related work has highlighted both the promise and limitations of quantum kernels. Early demonstrations of quantum-enhanced feature spaces focused on fidelity-style kernels induced by parametric circuits \cite{havlivcek2019supervised,schuld2019quantum}. Subsequent analyses emphasized that the apparent expressivity of large Hilbert spaces does not by itself guarantee learnability or generalization, and that kernel quality depends on data, encoding, and measurement choices \cite{huang2021power,thanasilp2024exponential}. We study two mitigation strategies that modify the similarity statistic through locality and scale mixing, and evaluate them using matched protocols across increasing feature dimension.

This work studies two practical strategies to mitigate the reliance on a single global overlap. Let $\mathcal{Q}_d=\{0,1,\dots,d-1\}$ denote the set of qubit indices. The first strategy constructs \emph{local (patch-wise) kernels} by evaluating similarity on small subsystems and aggregating the results. Let $\mathcal{P}=\{P_1,\dots,P_M\}$ be a collection of patches, with $P_m\subseteq\mathcal{Q}_d$, and let $P_m^{c}=\mathcal{Q}_d\setminus P_m$. For each patch $P_m$ we form the reduced state
\begin{equation}
\rho_{P_m}(\mathbf{x}) = \operatorname{Tr}_{P_m^{c}}\!\left(\ket{\psi(\mathbf{x})}\bra{\psi(\mathbf{x})}\right),
\label{eq:rdm_def}
\end{equation}
and define a patch kernel, for example via the Hilbert--Schmidt inner product
\begin{equation}
\kappa_{P_m}(\mathbf{x},\mathbf{x}') = \operatorname{Tr}\!\left(\rho_{P_m}(\mathbf{x})\,\rho_{P_m}(\mathbf{x}')\right).
\label{eq:hs_patch_kernel}
\end{equation}
The local kernel is obtained by a convex aggregation,
\begin{equation}
\begin{aligned}
k_{\mathrm{loc}}(\mathbf{x},\mathbf{x}')
&=
\sum_{m=1}^{M}
w_m\,\kappa_{P_m}(\mathbf{x},\mathbf{x}'),
\\
&w_m\geq 0,
\qquad
\sum_{m=1}^{M}w_m=1.
\end{aligned}
\label{eq:local_kernel}
\end{equation}

The second strategy constructs \emph{multi-scale kernels} by combining kernels computed at multiple patch granularities. Given a set of scales $\{\mathcal{P}^{(s)}\}_{s=1}^{S}$, we build one kernel per scale and define the multi-scale kernel function as a convex combination
\begin{equation}
\begin{aligned}
k_{\mathrm{ms}}(\mathbf{x},\mathbf{x}')
&=
\sum_{s=1}^{S}
\alpha_s\,k^{(s)}(\mathbf{x},\mathbf{x}'),
\\
&\alpha_s\geq 0,
\qquad
\sum_{s=1}^{S}\alpha_s=1.
\end{aligned}
\label{eq:multiscale_kernel}
\end{equation}
Here, $k^{(s)}$ denotes the kernel induced by scale $\mathcal{P}^{(s)}$.
Locality reduces sensitivity to global scrambling, while multi-scale mixing aims to preserve information that may appear at different subsystem sizes.

We evaluate the three kernel families under a matched and reproducible experimental protocol. We implement baseline (global), local, and multi-scale kernels in Qiskit under a unified API and benchmark them with matched experimental protocols over feature dimension $d\in\{4,6,\dots,20\}$.
To quantify kernel concentration and its downstream impact, we report off-diagonal summary statistics and effective rank as a summary of spectral diversity, use centered alignment with labels as a complementary diagnostic, and evaluate support vector machine (SVM) performance using precomputed kernels. Across the datasets we consider, local and multi-scale constructions consistently reshape kernel geometry relative to the baseline fidelity kernel, while improvements in accuracy remain dataset-dependent.

\noindent\textbf{Contributions and organization.} We provide (i) Qiskit implementations of baseline, local, and multi-scale quantum kernels under a unified API, (ii) a reproducible benchmark pipeline with fixed preprocessing, splits, and hyperparameter policies, and (iii) diagnostics that quantify concentration and downstream performance across a feature-dimension sweep. The remainder of this paper is organized as follows. Section~\ref{theory} recalls expressivity-induced concentration in global fidelity kernels and analyzes the mean raw similarity scales of global, local, and multi-scale constructions using independent Haar-random states as an analytical reference. Section~\ref{methods} describes the kernel constructions, feature maps, datasets, and evaluation protocol. Section~\ref{results} presents experimental results and diagnostic analyses. Code and configuration files to reproduce the experiments and figures are available in the accompanying repository.

\section{Theoretical analysis of concentration}\label{theory}

This section provides an analytical reference for how locality and multi-scale aggregation modify the mean similarity scales relevant to quantum-kernel concentration. We begin by recalling the expressivity-induced concentration behavior of global fidelity kernels established by \textit{Thanasilp et al.} \cite{thanasilp2024exponential}. We then use independent Haar-random pure states as a reference model to compare the mean raw global fidelity scale $2^{-d}$ with the mean raw Hilbert--Schmidt similarity scale $2^{-k}$ on a $k$-qubit patch. Finally, we show that local aggregation preserves this patch-dependent scale and that multi-scale aggregation produces convex combinations of the scales associated with its constituent patch sizes.

The local and multi-scale calculations below apply specifically to the Hilbert--Schmidt patch similarity defined in Eq.~\eqref{eq:hs_patch_kernel}. We do not assume that the data-encoded states used in our experiments are Haar-random. Independent Haar-random states are used only as a reference model to illustrate how the mean raw similarity depends on the dimension of the subsystem on which it is evaluated. The mean-scale identities concern raw kernel values before the unit-diagonal Gram-matrix normalization described in Section~\ref{methods}.

\subsection{Exponential concentration in global kernels}

Recent work by \textit{Thanasilp et al.} \cite{thanasilp2024exponential} established that quantum kernels can exhibit exponential concentration as the number of qubits increases. Specifically, concentration may arise from several mechanisms, including highly expressive embeddings, global measurements, entanglement, and noise. Under the corresponding assumptions, kernel values become increasingly difficult to distinguish across different input pairs, potentially reducing the information retained by the Gram matrix.

\begin{theorem}[Expressivity-induced concentration, adapted from Thanasilp \textit{et al.}]
\label{thm:global_concen}
Consider the global fidelity kernel defined in Eq.~\eqref{eq:fidelity_kernel}. Let $\mathbf{x}$ and $\mathbf{x}'$ be drawn independently from the same input distribution, and let $\mathcal{U}_{\phi}$ denote the induced ensemble of data-encoding unitaries. Under the assumptions of \cite{thanasilp2024exponential}, for any $\delta>0$,
\begin{equation}
\Pr_{\mathbf{x},\mathbf{x}'}\!\left[
\left|
k_{\mathrm{fid}}(\mathbf{x},\mathbf{x}')
-
\mathbb{E}_{\mathbf{x},\mathbf{x}'}\!
\left[
k_{\mathrm{fid}}(\mathbf{x},\mathbf{x}')
\right]
\right|
\geq \delta
\right]
\leq
\frac{
\Gamma_d\!\left(\varepsilon_{\mathcal{U}_{\phi}}\right)
}{
\delta^2
}.
\label{eq:expressivity_concentration}
\end{equation}
Here, $\delta$ is the deviation threshold, $\varepsilon_{\mathcal{U}_{\phi}}$ quantifies the distance of the encoding ensemble from the Haar reference, with smaller values corresponding to greater expressivity, and $\Gamma_d(\varepsilon_{\mathcal{U}_{\phi}})$ is the corresponding
dimension- and expressivity-dependent concentration factor for the fidelity kernel. If $\varepsilon_{\mathcal{U}_{\phi}}=\mathcal{O}(b^{-d})$ for some $b>1$, the bound decreases exponentially with $d$.
\end{theorem}

For the Haar-random reference model used in the remainder of this section, let $\ket{\psi}$ and $\ket{\varphi}$ be two independent Haar-random pure states in $(\mathbb{C}^{2})^{\otimes d}$. Their mean global fidelity follows from the standard Haar first-moment identity \cite{thanasilp2024exponential,mele2024introduction},
\begin{equation}
\mathbb{E}\!\left[
k_{\mathrm{fid}}(\psi,\varphi)
\right]
=
\mathbb{E}\!\left[
|\langle\psi|\varphi\rangle|^2
\right]
=
2^{-d}.
\label{eq:global_haar_mean}
\end{equation}
Here, the expectation is taken over the two independent Haar-random states. This identity specifies the characteristic mean similarity scale and should not by itself be interpreted as a concentration bound.

Theorem~\ref{thm:global_concen} bounds deviations from the mean under sufficiently expressive encoding ensembles, while Eq.~\eqref{eq:global_haar_mean} evaluates that mean in the Haar-random reference model. Together, they provide the relevant dimensional intuition: in a Haar-like regime, global fidelities can concentrate around a scale that decreases as $2^{-d}$. This motivates examining how the characteristic mean similarity changes when evaluated on lower-dimensional local subsystems.

\subsection{Patch-wise kernel mean scale}

Equation~\eqref{eq:global_haar_mean} identifies the mean global-fidelity scale $2^{-d}$ in the Haar-random reference model. To examine how this scale changes when the similarity is evaluated on a local subsystem, fix a patch $P_m\in\mathcal{P}$ of size $k=|P_m|$.

Following the reduced-state and Hilbert--Schmidt constructions defined in Eqs.~\eqref{eq:rdm_def} and~\eqref{eq:hs_patch_kernel}, for two pure states $\ket{\psi}$ and $\ket{\varphi}$ we write
\begin{equation}
\rho_{P_m}^{\psi}
=
\operatorname{Tr}_{P_m^{c}}
\left(\ket{\psi}\bra{\psi}\right),
\qquad
\rho_{P_m}^{\varphi}
=
\operatorname{Tr}_{P_m^{c}}
\left(\ket{\varphi}\bra{\varphi}\right),
\end{equation}
and
\begin{equation}
\kappa_{P_m}(\psi,\varphi)
=
\operatorname{Tr}
\left(
\rho_{P_m}^{\psi}
\rho_{P_m}^{\varphi}
\right).
\end{equation}

The patch-wise similarity is evaluated on a $2^k$-dimensional subsystem rather than on the full $2^d$-dimensional Hilbert space. The following result shows that, within the Haar-random reference model, the mean raw similarity scale is determined by the patch size $k$.

\begin{theorem}[Patch-wise mean scale]
\label{thm:patch-concen}
Let $P_m\in\mathcal{P}$ be a patch of size $k$. For two independent Haar-random pure states $\ket{\psi},\ket{\varphi}\in(\mathbb{C}^{2})^{\otimes d}$,
\begin{equation}
    \mathbb{E}\!\left[
    \kappa_{P_m}(\psi,\varphi)
    \right]
    =
    2^{-k}.
\label{eq:patch_haar_mean}
\end{equation}
\end{theorem}
A proof based on the Haar first-moment identity and the independence of the two states is provided in Appendix~\ref{app:proofs}.

\begin{coro}[Ratio of characteristic mean scales]
\label{coro:patch_global_ratio}
Combining Eq.~\eqref{eq:global_haar_mean} with Theorem~\ref{thm:patch-concen},
\begin{equation}
    \frac{
    \mathbb{E}\!\left[
    \kappa_{P_m}(\psi,\varphi)
    \right]
    }{
    \mathbb{E}\!\left[
    k_{\mathrm{fid}}(\psi,\varphi)
    \right]
    }
    =
    2^{d-k}.
\label{eq:patch_global_ratio}
\end{equation}
Thus, within the Haar-random reference model, replacing the global fidelity by a patch-wise Hilbert--Schmidt similarity changes the mean raw similarity scale from $2^{-d}$ to $2^{-k}$. For $k\ll d$, the mean patch similarity is larger than the mean global fidelity by a factor $2^{d-k}$.
\end{coro}

\subsection{Local kernel mean scale}

The previous result establishes the mean raw similarity scale for a single patch. We now consider the convex aggregation of patch-wise similarities defined in Eq.~\eqref{eq:local_kernel}. For two pure states $\ket{\psi}$ and $\ket{\varphi}$, we write $k_{\mathrm{loc}}(\psi,\varphi)$ for the corresponding state-level local kernel obtained by aggregating the patch similarities $\kappa_{P_m}(\psi,\varphi)$ introduced in the previous subsection.

\begin{theorem}[Local kernel mean scale]
\label{thm:local_concen}
Let $\mathcal{P}=\{P_1,\dots,P_M\}$ be a collection of patches of equal size $|P_m|=k$, and let $k_{\mathrm{loc}}$ be the local kernel defined by Eq.~\eqref{eq:local_kernel}, with fixed weights $w_m\geq 0$ satisfying $\sum_{m=1}^{M}w_m=1$. For two independent Haar-random pure states
$\ket{\psi},\ket{\varphi}\in(\mathbb{C}^{2})^{\otimes d}$,
\begin{equation}
    \mathbb{E}\!\left[
    k_{\mathrm{loc}}(\psi,\varphi)
    \right]
    =
    2^{-k}.
\label{eq:local_haar_mean}
\end{equation}
\end{theorem}

A proof based on Theorem~\ref{thm:patch-concen} and linearity of expectation is provided in Appendix~\ref{app:proofs}.

Theorem~\ref{thm:local_concen} shows that convex aggregation preserves the patch-wise mean raw similarity scale established in Theorem~\ref{thm:patch-concen}. When all patches have size $k$, the mean local similarity remains $2^{-k}$ rather than returning to the global-fidelity scale $2^{-d}$. This result concerns the mean raw similarity and does not by itself provide an additional concentration bound.

\subsection{Multi-scale kernel mean scale}

Different patch sizes probe different subsystem resolutions. Smaller patches emphasize local information and, within the Haar-random reference model, are associated with larger mean raw similarity scales. Larger patches access broader subsystems but approach the smaller global-fidelity scale as their size increases. Multi-scale kernels combine these subsystem resolutions rather than relying on a single patch size.

For each scale $s\in\{1,\dots,S\}$, let
\begin{equation}
    \mathcal{P}^{(s)}
    =
    \left\{
    P_1^{(s)},P_2^{(s)},\dots,P_{M_s}^{(s)}
    \right\}
\end{equation}
be a collection of $M_s$ patches satisfying $|P_m^{(s)}|=k_s$ for every $m$. For two pure states $\ket{\psi}$ and $\ket{\varphi}$, the corresponding scale-wise kernel is
\begin{equation}
    k^{(s)}(\psi,\varphi)
    =
    \frac{1}{M_s}
    \sum_{m=1}^{M_s}
    \kappa_{P_m^{(s)}}(\psi,\varphi).
\label{eq:theory_scale_kernel}
\end{equation}
This gives the explicit state-level form of the scale kernel $k^{(s)}$ appearing in the multi-scale construction defined in Eq.~\eqref{eq:multiscale_kernel}. When $k_s=d$, the Hilbert--Schmidt similarity on the full system coincides with the fidelity for pure states.

\begin{lemma}[Scale-wise mean scale]
\label{lem:scale_concen}
For two independent Haar-random pure states
$\ket{\psi},\ket{\varphi}\in(\mathbb{C}^{2})^{\otimes d}$,
the scale-wise kernel in Eq.~\eqref{eq:theory_scale_kernel} satisfies
\begin{equation}
    \mathbb{E}\!\left[
    k^{(s)}(\psi,\varphi)
    \right]
    =
    2^{-k_s}.
\label{eq:scale_haar_mean}
\end{equation}
\end{lemma}

Lemma~\ref{lem:scale_concen} follows directly from Theorem~\ref{thm:patch-concen} and linearity of expectation, since every patch at scale $s$ has size $k_s$.

\begin{theorem}[Multi-scale kernel mean scale]
\label{thm:multiscale_concen}
Let $k_{\mathrm{ms}}$ be the multi-scale kernel defined in Eq.~\eqref{eq:multiscale_kernel}, with fixed weights $\alpha_s\geq0$ satisfying $\sum_{s=1}^{S}\alpha_s=1$. For two independent Haar-random pure states $\ket{\psi},\ket{\varphi}\in(\mathbb{C}^{2})^{\otimes d}$,
\begin{equation}
    \mathbb{E}\!\left[
    k_{\mathrm{ms}}(\psi,\varphi)
    \right]
    =
    \sum_{s=1}^{S}
    \alpha_s\,2^{-k_s}.
\label{eq:multiscale_haar_mean}
\end{equation}
\end{theorem}

A proof based on Lemma~\ref{lem:scale_concen} and linearity of expectation is provided in Appendix~\ref{app:proofs}.

\begin{coro}[Convexity bound]
\label{coro:bound}
Under the assumptions of Theorem~\ref{thm:multiscale_concen}, let $k_{\min}=\min_s k_s$ and $k_{\max}=\max_s k_s$. Then
\begin{equation}
    2^{-k_{\max}}
    \leq
    \mathbb{E}\!\left[
    k_{\mathrm{ms}}(\psi,\varphi)
    \right]
    \leq
    2^{-k_{\min}}.
\label{eq:multiscale_convexity_bound}
\end{equation}
\end{coro}
The result follows because $\mathbb{E}[k_{\mathrm{ms}}]$ is a convex combination of the values $\{2^{-k_s}\}_{s=1}^{S}$.

Theorem~\ref{thm:multiscale_concen} and Corollary~\ref{coro:bound} show that the mean raw similarity scale of the multi-scale kernel is a convex combination of the scales associated with its constituent patch sizes. It therefore lies between the scales corresponding to the smallest and largest patches and can incorporate contributions from multiple subsystem resolutions. These results concern mean raw similarities and do not by themselves determine the variance, similarity distribution, kernel spectrum, label alignment, or predictive performance.

\section{Methods}\label{methods}

\subsection{Kernel definitions}\label{sec:kernel_definitions}

Throughout, we represent a classical input as a feature vector $\mathbf{x}\in\mathbb{R}^d$. For a dataset $\{\mathbf{x}_i\}_{i=1}^n$, a kernel function $k(\cdot,\cdot)$ induces a Gram matrix $K\in\mathbb{R}^{n\times n}$ with entries $K_{ij}=k(\mathbf{x}_i,\mathbf{x}_j)$ as in Eq.~\eqref{eq:gram_def}. The corresponding encoded quantum state $\ket{\psi(\mathbf{x})}$ is prepared by applying the feature-map circuit $U_{\phi}(\mathbf{x})$ as in Eq.~\eqref{eq:feature_state}.

\subsubsection{Baseline (global fidelity) kernel}

The baseline kernel is the global fidelity kernel defined in Eq.~\eqref{eq:fidelity_kernel}. Operationally, we compute it by first preparing the statevector for each sample,
\begin{equation}
\mathbf{s}_i \equiv \ket{\psi(\mathbf{x}_i)} \in \mathbb{C}^{2^d},
\label{eq:baseline_statevector_def}
\end{equation}
then forming the overlap matrix
\begin{equation}
G_{ij} = \langle \psi(\mathbf{x}_i) \mid \psi(\mathbf{x}_j) \rangle,
\qquad
G = S S^{\dagger},
\label{eq:baseline_overlap_def}
\end{equation}
where $S\in\mathbb{C}^{n\times 2^d}$ is the matrix whose $i$th row is $\mathbf{s}_i^{\dagger}$ (equivalently, $S=[\mathbf{s}_1^{\dagger};\dots;\mathbf{s}_n^{\dagger}]$). Finally, the Gram matrix is given entrywise by
\begin{equation}
K^{\mathrm{fid}}_{ij} = \left|G_{ij}\right|^2.
\label{eq:baseline_gram_def}
\end{equation}

In our implementation we symmetrize $K^{\mathrm{fid}}$ by replacing it with $(K^{\mathrm{fid}}+(K^{\mathrm{fid}})^{\top})/2$ and set the diagonal entries to $1$. The symmetrization enforces the exact symmetry that holds in the ideal definition, but that may be violated slightly by finite-precision arithmetic; this improves numerical stability for downstream operations such as spectral decompositions. Setting $K^{\mathrm{fid}}_{ii}=1$ enforces the exact self-similarity $|\langle\psi(\mathbf{x}_i)\mid\psi(\mathbf{x}_i)\rangle|^2=1$, which can otherwise deviate from unity at the level of floating-point error.

\subsubsection{Local (patch-wise) kernels}

The motivation for locality is that exponential concentration of the global fidelity kernel is driven by a single overlap of $d$-qubit states, which can become nearly constant (and typically small) for most distinct inputs as circuit expressivity grows. By instead evaluating similarities on small subsystems and aggregating them, local kernels probe overlaps in smaller effective Hilbert spaces and are less sensitive to global scrambling. As a result, the aggregated Gram matrix can retain richer off-diagonal variation than the baseline fidelity kernel while still defining a valid similarity measure.

Let $\mathcal{P}=\{P_1,\dots,P_M\}$ be a collection of patches, where each patch $P_m\subseteq\mathcal{Q}_d$ indexes a subset of qubits. When not specified otherwise, we use the default disjoint partition into adjacent pairs $P_m=\{2m-2,2m-1\}$ for $m=1,\dots,d/2$. When specified partitions are provided, we set $\mathcal{P}$ equal to that collection of index sets and apply the same patch-kernel construction and aggregation without further modification. For completeness, Algorithm~\ref{alg:local_kernel} in Appendix~\ref{app:local_kernel_algo} summarizes the end-to-end procedure we use to construct the local (patch-wise) kernel.

We support two implementations for computing patch-wise similarities (followed by a common aggregation step).

\paragraph{Subcircuit-based patch kernel.}
For each patch $P_m$, we construct a feature-map circuit restricted to the qubits in $P_m$ and use the corresponding feature subvector $\mathbf{x}_{P_m}$ obtained by selecting the components of $\mathbf{x}$ indexed by $P_m$. Denoting the resulting patch state by $\ket{\psi_{P_m}(\mathbf{x})}$, the patch kernel is
\begin{equation}
 k^{(P_m)}_{\mathrm{sub}}(\mathbf{x},\mathbf{x}') = \left|\langle\psi_{P_m}(\mathbf{x})\mid\psi_{P_m}(\mathbf{x}')\rangle\right|^2.
 \label{eq:local_subcircuit_kernel}
\end{equation}

\paragraph{Reduced Density Matrix (RDM) local kernel.}
Alternatively, we prepare the full $d$-qubit state $\ket{\psi(\mathbf{x})}$ and compute the reduced density matrix on a patch by tracing out the complement,
\begin{equation}
 \rho_{P_m}(\mathbf{x}) = \operatorname{Tr}_{P_m^{c}}\!\left(\ket{\psi(\mathbf{x})}\bra{\psi(\mathbf{x})}\right).
 \label{eq:local_rdm_def}
\end{equation}
Here, $P_m^{c}=\mathcal{Q}_d\setminus P_m$. We then define the patch similarity using either the quantum state fidelity $F\!\left(\rho_{P_m}(\mathbf{x}),\rho_{P_m}(\mathbf{x}')\right)$ or the Hilbert--Schmidt inner product $\operatorname{Tr}\!\left(\rho_{P_m}(\mathbf{x})\,\rho_{P_m}(\mathbf{x}')\right)$.

\paragraph{Patch aggregation.}
Patch kernels are aggregated by either an unweighted mean or a weighted mean. In the unweighted case, the arithmetic mean corresponds to using uniform weights $w_m=1/M$. With weights $w_m\ge 0$ and $\sum_{m=1}^M w_m=1$, the aggregated local kernel is
\begin{equation}
 k_{\mathrm{loc}}(\mathbf{x},\mathbf{x}') = \sum_{m=1}^{M} w_m\,k^{(P_m)}(\mathbf{x},\mathbf{x}').
 \label{eq:local_agg_kernel}
\end{equation}

\subsubsection{Multi-scale kernels}

A multi-scale kernel combines kernels computed at multiple granularities. The motivation is that relevant structure may appear at different subsystem sizes: small patches can capture short-range correlations, while larger patches can capture more global similarity. By mixing multiple patch granularities, the kernel can reduce reliance on a single highly concentrated global overlap while retaining informative similarities at intermediate scales.

We specify a collection of scales $\{\mathcal{P}^{(s)}\}_{s=1}^{S}$, where each scale $\mathcal{P}^{(s)}$ is a collection of patches. For each scale, we compute the scale-wise kernel $k^{(s)}$ by averaging the corresponding patch similarities. The final multi-scale kernel is the convex combination defined in Eq.~\eqref{eq:multiscale_kernel}.

Before constructing this combination, the nonnegative scale weights $\{\alpha_s\}$ are renormalized to sum to one. Unless specified otherwise, we use uniform weights, $\alpha_s=1/S$; in particular, the default two-scale construction uses $\alpha_1=\alpha_2=1/2$. When scales are not specified, the default construction consists of (i) disjoint adjacent pairs and (ii) the full system.

In our implementation, we compute one kernel per scale by averaging patch contributions, then combine the per-scale kernels with nonnegative weights that are renormalized to sum to one. For patches smaller than the full system we use the Hilbert--Schmidt inner product between patch reduced density matrices, while for the full-system patch we use the state fidelity. Finally, unless otherwise specified, we normalize the resulting multi-scale Gram matrix to have unit diagonal and enforce exact symmetry.
For completeness, Algorithm~\ref{alg:multiscale_kernel} in Appendix~\ref{app:multiscale_kernel_algo} summarizes the end-to-end procedure we use to construct the multi-scale kernel.

\subsubsection{Normalization and PSD correction}

Kernel matrices computed from finite-precision statevector arithmetic can exhibit small asymmetries and, in some cases, tiny negative eigenvalues that are numerical artifacts rather than physical effects. To make kernels comparable across feature dimensions and constructions, we enforce a common \emph{unit-diagonal normalization}. For any kernel matrix $K$ with strictly positive diagonal entries, we apply
\begin{equation}
 K_{ij} \leftarrow \frac{K_{ij}}{\sqrt{K_{ii}K_{jj}}},
 \label{eq:kernel_unit_diag}
\end{equation}
which rescales the kernel so that $K_{ii}=1$ for all samples.

In addition, we explicitly symmetrize the numerically computed Gram matrices via
$K\leftarrow (K+K^{\top})/2$.

\paragraph{Local kernel.}
For the local (patch-wise) kernel, we further apply a positive-semidefinite (PSD) correction after aggregation by projecting onto the PSD cone using eigenvalue clipping. Concretely, for the eigendecomposition $K=V\,\mathrm{diag}(\lambda)\,V^{\top}$, we set negative eigenvalues below a small threshold to zero and reconstruct $K$. We then reapply the unit-diagonal normalization.

\paragraph{Baseline and multi-scale kernels.}
For the baseline fidelity kernel and the multi-scale kernel, our implementation enforces symmetry and a unit diagonal but does not perform an explicit PSD projection; in practice, any observed PSD violations in these constructions are typically at the level of numerical noise.

\subsection{Feature maps and implementation}

\paragraph{ZZ-style feature maps.}
All quantum kernels in this work are instantiated using ZZ-style data-encoding circuits that map an input feature vector $\mathbf{x}\in\mathbb{R}^{d}$ to a $d$-qubit state $\ket{\psi(\mathbf{x})}$. We implement three closely related variants (selected via a unified API): (i) a simple \emph{manual} ZZ-style map (\texttt{zz\_manual}) based on single-qubit $R_x$ rotations followed by $CZ$ entangling gates; (ii) a more \emph{canonical-like manual} ZZ map (\texttt{zz\_manual\_canonical}) based on an initial layer of Hadamards followed by repeated layers of local $R_z$ rotations and pairwise $R_{ZZ}$ interactions; and (iii) Qiskit's circuit-library implementation (\texttt{zz\_qiskit}) based on \texttt{qiskit.circuit.library.zz\_feature\_map}.

\paragraph{Depth and entanglement patterns.}
We denote the number of repetitions (layers) of the feature map by $L$ (\texttt{depth} in the code). Unless otherwise specified, we use shallow feature maps (typically $L=1$) to enable matched sweeps over feature dimension $d\in\{4,6,\dots,20\}$. For the manual feature maps, we use \texttt{linear} and \texttt{ring} entanglement patterns (nearest-neighbor couplings, with \texttt{ring} additionally coupling the last and first qubits). For the Qiskit feature map, we use entanglement patterns supported by the library implementation (we use \texttt{linear} in our main sweeps).

\paragraph{Statevector backend.}
All kernel matrices reported in this paper are computed using exact statevector simulation in Qiskit \cite{qiskit}. Concretely, each kernel implementation constructs the feature-map circuit for each sample and evaluates state overlaps (baseline and full-patch contributions) via inner products of statevectors. Local and multi-scale kernels compute patch reduced density matrices using partial traces on the statevector and then evaluate patch similarities.

\paragraph{Kernel post-processing: symmetry, normalization, and centering.}
Across kernel families we apply consistent numerical post-processing. First, we explicitly enforce symmetry via $K\leftarrow (K+K^{\top})/2$. Second, we apply unit-diagonal normalization as in Eq.~\eqref{eq:kernel_unit_diag}. For the local kernel only, we additionally apply eigenvalue clipping to enforce positive semidefiniteness, followed by a second unit-diagonal normalization.

Some experiments and diagnostics use \emph{centered} kernels. Given a Gram matrix $K\in\mathbb{R}^{n\times n}$, we center it as $K_c = HKH$ with $H=I-\frac{1}{n}\mathbf{1}\mathbf{1}^{\top}$. Unless otherwise specified, we report results for uncentered kernels and compute centered alignment with labels as a diagnostic.

\subsection{Datasets and preprocessing}

\paragraph{Datasets.}
We evaluate on a mixture of synthetic and real-world datasets, using a common pipeline that produces matched train/validation/test splits and precomputed-kernel SVM evaluation. The main experiments use \texttt{breast\_cancer} (via \texttt{sklearn.datasets.load\_breast\_cancer}) \cite{sklearn_breast_cancer}, \texttt{parkinsons} (via \texttt{sklearn.datasets.fetch\_openml} with \texttt{data\_id=1488}) \cite{openml_1488_parkinsons}, \texttt{ionosphere} \cite{kaggle_ionosphere}, \texttt{heart\_disease} \cite{uci_heart_disease}, and subset runs for larger datasets such as \texttt{exam\_score\_prediction} \cite{kaggle_exam_score_prediction} and \texttt{star\_classification} \cite{kaggle_star_classification}.

\paragraph{Scaling and angle encoding.}
All datasets are standardized prior to quantum embedding. After scaling, each feature is treated as a rotation angle and passed directly to the selected ZZ-style feature map, so that the feature dimension equals the number of qubits ($d$).

\paragraph{Target feature dimensions $d\in\{4,6,\dots,20\}$.}
To study concentration as a function of system size, we evaluate each dataset at a common set of feature dimensions $d\in\{4,6,\dots,20\}$. When a dataset has more than $d$ available features, we reduce to the target dimension using PCA.
When a dataset has fewer raw features than required for the sweep, we augment it with simple engineered features based on pairwise interactions (products) to reach the desired dimension while keeping the protocol consistent across kernels.

\paragraph{Splits and random seeds.}
For each dataset and target dimension $d$, we evaluate multiple deterministic train/validation/test splits generated from a fixed grid of random seeds. For each seed, the same split, preprocessing, and hyperparameter grid are used across all kernel families. Results aggregated across seeds report the mean, with variability indicated by one standard deviation where shown. All split indices and labels are saved alongside each kernel matrix for reproducibility.

\subsection{Scalability: Nyström approximation (optional)}

Computing full kernel matrices scales quadratically in the number of samples ($O(n^2)$) and can become a bottleneck for larger datasets. As an optional scalability mechanism, our codebase supports a Nystr\"om/landmark approximation, which approximates the full Gram matrix using a smaller set of $m\ll n$ landmark points \cite{williams2001nystrom}.

\paragraph{Approximation.}
Let $X=\{\mathbf{x}_i\}_{i=1}^{n}$ be the full dataset and let $Z=\{\mathbf{z}_j\}_{j=1}^{m}\subset X$ denote a
subset of landmark points. Define the cross-kernel matrix $C\in\mathbb{R}^{n\times m}$ with entries $C_{ij}=k(\mathbf{x}_i,\mathbf{z}_j)$ and the landmark Gram matrix $W\in\mathbb{R}^{m\times m}$ with entries $W_{ij}=k(\mathbf{z}_i,\mathbf{z}_j)$. The Nystr\"om approximation of the full kernel is then
\begin{equation}
\tilde{K} = C\,W^{\dagger}\,C^{\top},
\label{eq:nystrom_kernel}
\end{equation}
where $W^{\dagger}$ denotes the (pseudo-)inverse of $W$. Equivalently, one can form an explicit feature representation
$\Phi = C\,(W^{\dagger})^{1/2}$ and train a linear model on $\Phi$.

\paragraph{Implementation details.}
We implement Nystr\"om support by providing cross-kernel computation routines for each kernel family. These routines compute the cross-kernel $C=K(X,Z)$ between a set of samples and a set of landmark points, optionally using chunking to avoid large intermediate allocations, and they reuse the same patch/scale logic as the corresponding full-kernel construction. On the evaluation side, our SVM utility supports both (i) precomputed-kernel SVMs for full kernels and (ii) linear SVMs on explicit feature matrices for Nystr\"om-style approximations.

\paragraph{Why we include it.}
Nystr\"om provides a standard and lightweight route to scaling kernel experiments while keeping the kernel construction tied to the same underlying similarity function. In this project it primarily serves as infrastructure for future experiments on larger datasets; unless noted otherwise, the results in this paper use the exact (non-Nystr\"om) statevector kernels.

\subsection{Diagnostics and evaluation}

To quantify how strongly a kernel concentrates and how much information it retains for downstream learning, we compute three geometry diagnostics computed from the Gram matrix $K\in\mathbb{R}^{n\times n}$, together with an SVM evaluation using precomputed kernels.

\paragraph{Off-diagonal concentration (p50/p95).}
Let $\mathcal{I}=\{(i,j): i\neq j\}$ denote the off-diagonal index set and let $\{K_{ij}\}_{(i,j)\in\mathcal{I}}$ be the multiset of off-diagonal entries. We summarize concentration using the median (p50) and the upper-tail percentile (p95),
\begin{equation}
\begin{split}
\mathrm{p50}(K) = \mathrm{percentile}_{50}\!\left(\{K_{ij}\}_{i\neq j}\right),
\\
\mathrm{p95}(K) = \mathrm{percentile}_{95}\!\left(\{K_{ij}\}_{i\neq j}\right).
\end{split}
\label{eq:offdiag_percentiles}
\end{equation}
As $d$ (or circuit expressivity) increases, exponential concentration manifests as off-diagonal entries shrinking toward zero; correspondingly, both p50 and p95 decrease and the kernel approaches the identity.

\paragraph{Effective rank (entropy-based).}
To characterize spectral richness, we compute the eigenvalues of the symmetrized kernel $K_s = \tfrac{1}{2}(K+K^{\top})$ and denote them by $\{\lambda_i\}_{i=1}^{n}$. We clip small negative eigenvalues (numerical artifacts) to $0$ and form a probability distribution $p_i = \lambda_i / \sum_j \lambda_j$. The entropy-based effective rank is then \cite{roy2007effective}
\begin{equation}
 r_{\mathrm{eff}}(K)
 = \exp\!\left(-\sum_{i=1}^{n} p_i\,\log p_i\right).
 \label{eq:effective_rank_entropy}
\end{equation}
Here, terms with $p_i=0$ are defined by the convention
$0\log 0=0$. A large $r_{\mathrm{eff}}$ indicates a flatter normalized eigenvalue distribution, whereas a small $r_{\mathrm{eff}}$ indicates that spectral mass is concentrated in fewer dominant eigenvalues. Effective rank is a spectral summary and, by itself, does not measure kernel quality or label relevance; in particular, the identity matrix has maximal effective rank.

\paragraph{Centered alignment with labels.}
We compute the centered kernel alignment between $K$ and a label kernel $L$ \cite{cristianini2002kernel}. Let $H = I - \tfrac{1}{n}\mathbf{1}\mathbf{1}^{\top}$ be the centering matrix, and define centered kernels $K_c = HKH$ and $L_c = HLH$.
For multi-class labels $y\in\{1,\dots,C\}^n$, we build $L$ from one-hot encodings $Y\in\mathbb{R}^{n\times C}$ via $L = YY^{\top}$. The centered alignment is
\begin{equation}
\mathcal{A}(K,y)
= \frac{\langle K_c, L_c\rangle_F}{\lVert K_c\rVert_F\,\lVert L_c\rVert_F},
\qquad
\langle A,B\rangle_F = \operatorname{Tr}(A^{\top}B).
\label{eq:centered_alignment}
\end{equation}
This statistic measures how well the similarity structure in $K$ matches label similarity after removing global mean effects.

\paragraph{SVM protocol and hyperparameter grid.}
For each dataset (and each target dimension $d$), we use fixed train/validation/test splits that are shared across kernel families. Given a full Gram matrix $K$, we train an SVM with a precomputed kernel using $K_{\mathrm{train}} = K[\mathrm{train},\mathrm{train}]$. For model selection, we sweep the regularization parameter $C$ over the grid
\begin{equation}
 \mathcal{C} = \{0.1,\; 1,\; 10\},
 \label{eq:svm_C_grid}
\end{equation}
choose the value maximizing validation accuracy, and then report test accuracy using the selected $C$ on
$K_{\mathrm{test}} = K[\mathrm{test},\mathrm{train}]$. All kernels are evaluated under the same splits and the same $\mathcal{C}$ grid to enable matched comparisons.

\section{Results and analysis}\label{results}

\subsection{Kernel concentration and effective rank vs.~feature dimension}
\label{sec:results_geometry_vs_d}

We first study how kernel geometry changes as the feature dimension $d$ increases. For each dataset we run a matched sweep over $d\in\{4,6,\dots,20\}$ and compute (i) off-diagonal concentration statistics (median p50 and upper-tail p95) and (ii) spectral diversity via the entropy-based effective rank $r_{\mathrm{eff}}$ defined in Eq.~\eqref{eq:effective_rank_entropy}. All kernels are normalized to unit diagonal, so changes in p50/p95 directly reflect how much similarity mass remains in the off-diagonals as $d$ grows. Unless otherwise noted, all kernel families are evaluated using matched train/validation/test splits for each seed, following the protocol described in Section~\ref{methods}, with aggregation over the fixed seed grid specified in the corresponding figure captions.

\paragraph{Off-diagonal concentration (p50).}
Figure~\ref{fig:offdiag_p50_vs_d} shows the off-diagonal median (p50) as a function of $d$. Across all datasets, the baseline (global fidelity) kernel exhibits the strongest concentration signature: as $d$ increases, the p50 curve rapidly decreases toward $0$, consistent with the Gram matrix approaching the identity. In contrast, the local kernel maintains substantially higher p50 values over the same sweep, indicating that patch-wise aggregation retains nontrivial similarity structure at larger $d$. As expected, the multi-scale kernel typically interpolates between the baseline and local curves because it explicitly mixes global and local similarity. In most cases the baseline--local gap widens with $d$, highlighting that locality preserves meaningful off-diagonal similarity even in regimes where the global fidelity kernel is close to an identity matrix.

\begin{figure*}[t]
\centering
\includegraphics[width=0.88\linewidth]{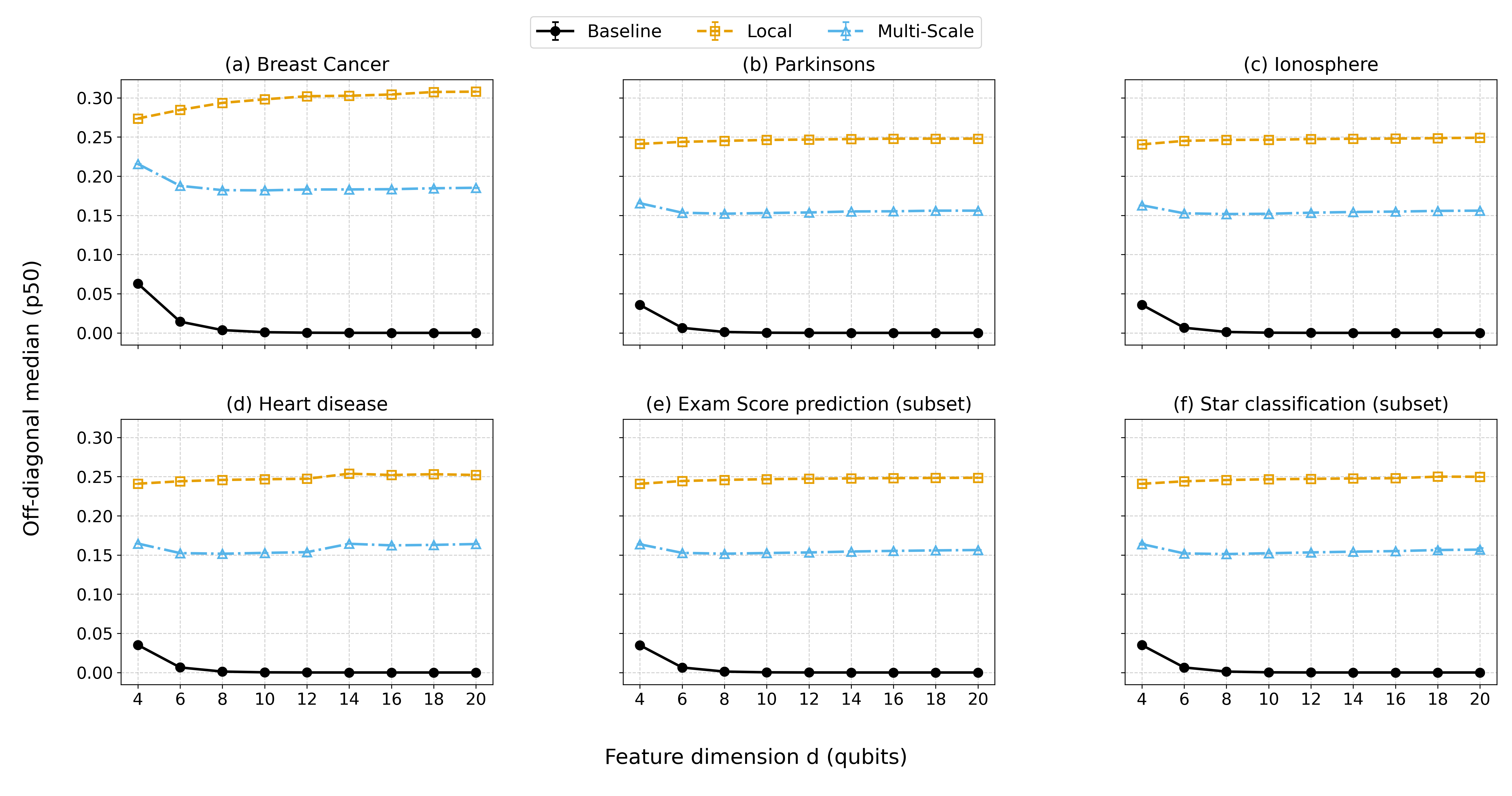}
\vspace{-5mm}
\caption{Median off-diagonal similarity (p50) vs.\ feature dimension $d$ for all datasets. The p50 values of the global fidelity baseline decrease rapidly toward zero as $d$ increases, whereas the local kernel retains substantially larger off-diagonal similarities and the multi-scale kernel is typically intermediate. All Gram matrices are normalized to unit diagonal.}
\label{fig:offdiag_p50_vs_d}
\end{figure*}

\paragraph{Upper-tail behavior (p95).}
The p95 statistic (Appendix~\ref{app:supplementary_figures}, Fig.~\ref{fig:offdiag_p95_vs_d}) probes whether a small subset of pairs remains highly similar even when the median similarity collapses. We again observe a consistent ordering: baseline concentrates fastest, local concentrates slowest, and multi-scale is intermediate. Operationally, this suggests that locality mitigates not only the typical pairwise overlap (p50) but also the collapse of the similarity tail (p95), thereby preserving more heterogeneous pairwise structure.

\paragraph{Effective rank.}
Figure~\ref{fig:effrank_vs_d} reports the entropy-based effective rank of the kernel spectrum. As $d$ increases, the baseline kernel consistently exhibits the highest $r_{\mathrm{eff}}$, reflecting a flatter normalized eigenvalue distribution. In contrast, the local kernel exhibits a lower $r_{\mathrm{eff}}$, whereas the multi-scale kernel tends to lie between the baseline and local constructions. This ordering differs from that observed for the off-diagonal statistics, showing that effective rank and off-diagonal concentration capture complementary aspects of the kernel geometry. In this setting, a larger $r_{\mathrm{eff}}$ indicates that spectral mass is distributed more uniformly across the eigenvalue spectrum, while a smaller $r_{\mathrm{eff}}$ indicates that spectral mass is concentrated in fewer dominant eigenvalues. Although $r_{\mathrm{eff}}$ is a coarse summary, it captures spectral diversity rather than kernel quality.

\paragraph{Representative eigen-spectra (qualitative).}
In addition to the effective-rank summary, Appendix~\ref{app:supplementary_figures} (Fig.~\ref{fig:spectra_examples_d12}) provides representative eigen-spectrum examples at a fixed feature dimension ($d=12$). A consistent qualitative signature is that the baseline kernel exhibits a flatter eigenvalue spectrum, whereas the local kernel allocates a larger fraction of the spectral mass to the leading eigenvalues, with the multi-scale kernel exhibiting intermediate behavior. This pattern is consistent with the effective-rank values reported in Fig.~\ref{fig:effrank_vs_d} and further illustrates that effective rank characterizes spectral diversity rather than kernel concentration.

\begin{figure*}[t]
\centering
\includegraphics[width=0.88\linewidth]{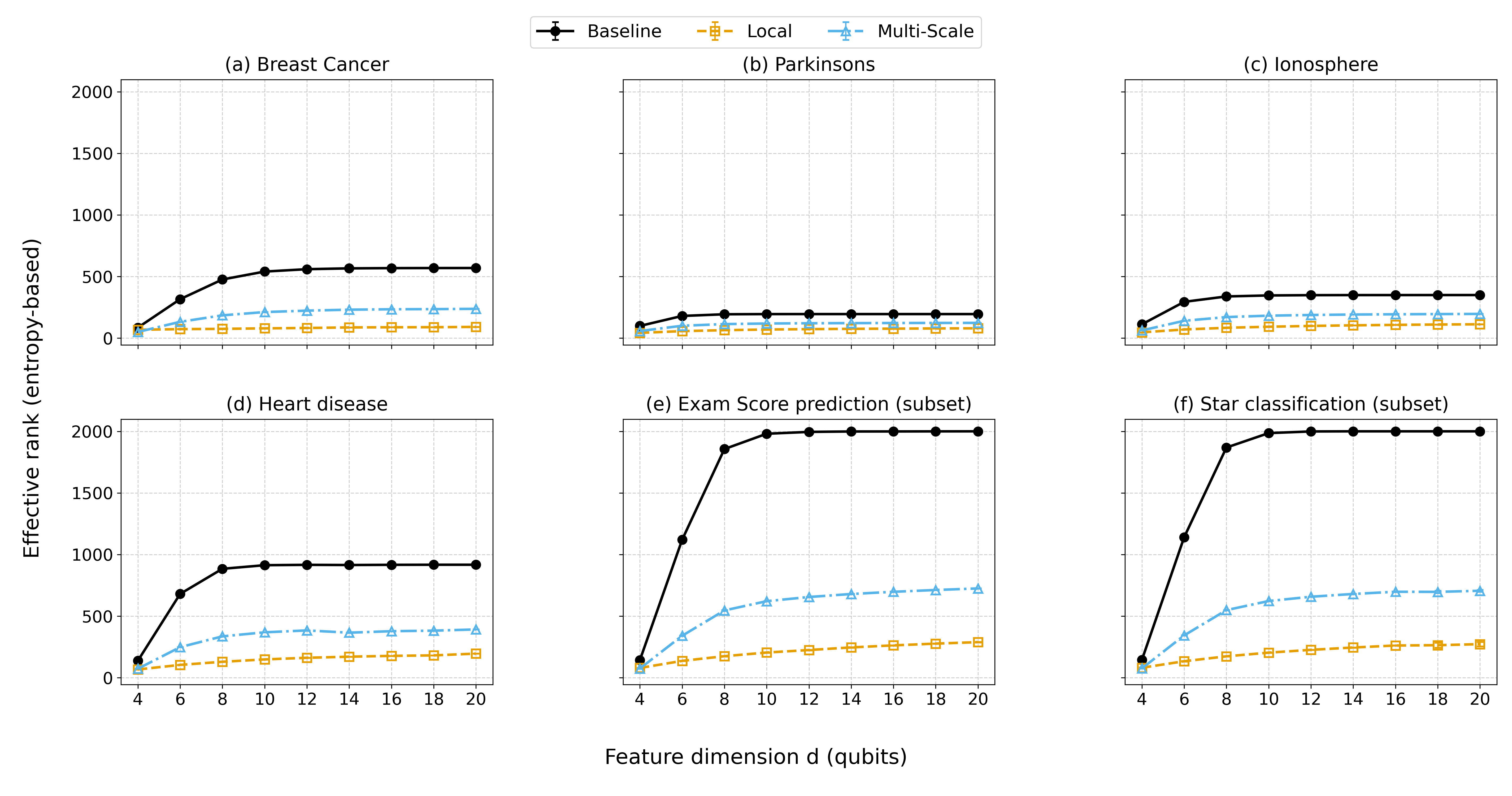}
\vspace{-5mm}
\caption{Entropy-based effective rank $r_{\mathrm{eff}}(K)$ vs.\ feature dimension $d$. The global fidelity baseline exhibits the highest effective rank, corresponding to a flatter normalized eigenvalue distribution. The local kernel generally exhibits lower values, while the multi-scale kernel typically lies between the baseline and local constructions.}
\label{fig:effrank_vs_d}
\end{figure*}

\subsection{SVM performance}
\label{sec:results_svm}

We next evaluate whether the geometric changes induced by local and multi-scale kernels translate into improved downstream classification. For each dataset and dimension $d$, we train an SVM with a precomputed kernel on the training split. The regularization parameter $C$ is selected by validation accuracy from the fixed grid $\mathcal{C}=\{0.1,1,10\}$ (Eq.~\eqref{eq:svm_C_grid}) independently for each kernel family; we then report the corresponding test accuracy.
Figure~\ref{fig:bestC_vs_d} in Appendix~\ref{app:supplementary_figures} summarizes the validation-selected $C$ values as a function of $d$. This protocol isolates the effect of the kernel while keeping model selection simple and reproducible.

Figure~\ref{fig:testacc_vs_d} shows test accuracy as a function of $d$. Across datasets, we observe that reducing concentration (Section~\ref{sec:results_geometry_vs_d}) does not necessarily imply improved accuracy: local and multi-scale kernels can match or exceed the baseline on some datasets and dimensions, while remaining comparable or worse on others. This highlights that concentration is a useful diagnostic for kernel geometry, but not a sufficient criterion for predictive performance without additional choices (e.g., patch design, feature map, or depth). Mechanistically, locality and multi-scale mixing change the inductive bias of the similarity measure: while they can mitigate global-state overlap collapse, they can also discard global correlations (or dilute them via averaging) that may be important for a given dataset.

\begin{figure*}[t]
\centering
\includegraphics[width=0.88\linewidth]{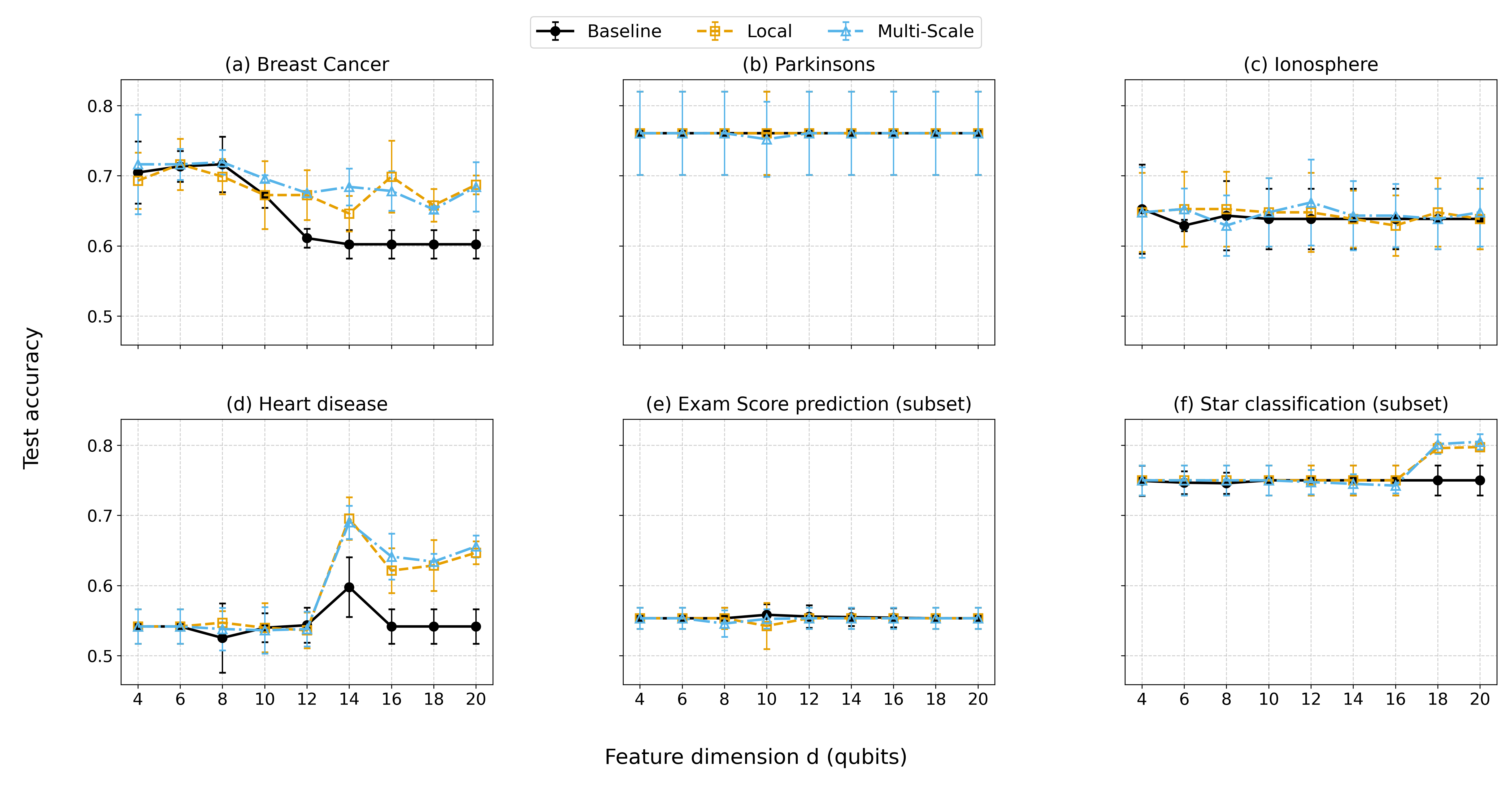}
\vspace{-5mm}
\caption{SVM test accuracy vs.\ feature dimension $d$. Points show the mean test accuracy across the fixed seed grid, and error bars show the standard deviation. For each dataset, dimension, seed, and kernel family, the SVM regularization parameter $C$ is selected by validation accuracy from the fixed grid in Eq.~\eqref{eq:svm_C_grid}. The selected value is then used for one evaluation on the corresponding test split. Differences between the baseline, local, and multi-scale kernels are dataset- and dimension-dependent.}
\label{fig:testacc_vs_d}
\end{figure*}

\subsection{Tradeoffs and discussion}
\label{sec:results_tradeoffs}

The results above show a consistent geometric effect (local and multi-scale constructions mitigate concentration relative to the baseline fidelity kernel), but the performance impact is more nuanced. In particular, reduced concentration does \emph{not} imply higher accuracy in a dataset-independent way.

\paragraph{Reduced concentration does not necessarily imply higher accuracy.}
A kernel that is less concentrated (higher off-diagonal p50/p95) can still fail to improve test accuracy if the preserved variation is not aligned with the label structure. This highlights that kernel concentration and effective rank capture different geometric properties. Conversely, a more concentrated kernel can remain competitive when the classification task is simple under the chosen encoding, or when regularization compensates for limited similarity variation.

\paragraph{Alignment vs.\ accuracy.}
To bridge geometry and performance, we compute the centered alignment statistic $\mathcal{A}(K,y)$ (Eq.~\eqref{eq:centered_alignment}). While alignment is not a direct proxy for SVM accuracy, it provides a complementary diagnostic: kernels that increase off-diagonal mass but do not improve (or even degrade) alignment can reshape geometry without adding label-relevant structure. Intuitively, centered alignment measures whether variations in $K$ co-vary with label similarity after removing global mean effects; thus, it is sensitive to whether the additional variability retained by local or multi-scale constructions is plausibly task-relevant rather than merely ``less concentrated.''

\paragraph{Computational tradeoffs.}
Local and multi-scale kernels reduce concentration by replacing a single global overlap with multiple subsystem-level comparisons, but this typically increases computational overhead. In particular, RDM-based patch kernels require additional partial traces (one per patch and sample in the simplest implementation), and multi-scale kernels further multiply this cost across scales; these tradeoffs motivate the optional Nystr\"om approximation (Section~\ref{methods}) when scaling to larger datasets. 

\paragraph{Global view: test-accuracy deltas.}
Figure~\ref{fig:delta_testacc_heatmaps} summarizes the test-accuracy change relative to the baseline kernel as
a function of dataset and $d$. This view makes clear that gains from local and multi-scale kernels are heterogeneous across datasets and across dimensions.

\begin{figure*}[t]
\centering
\includegraphics[width=0.88\linewidth]{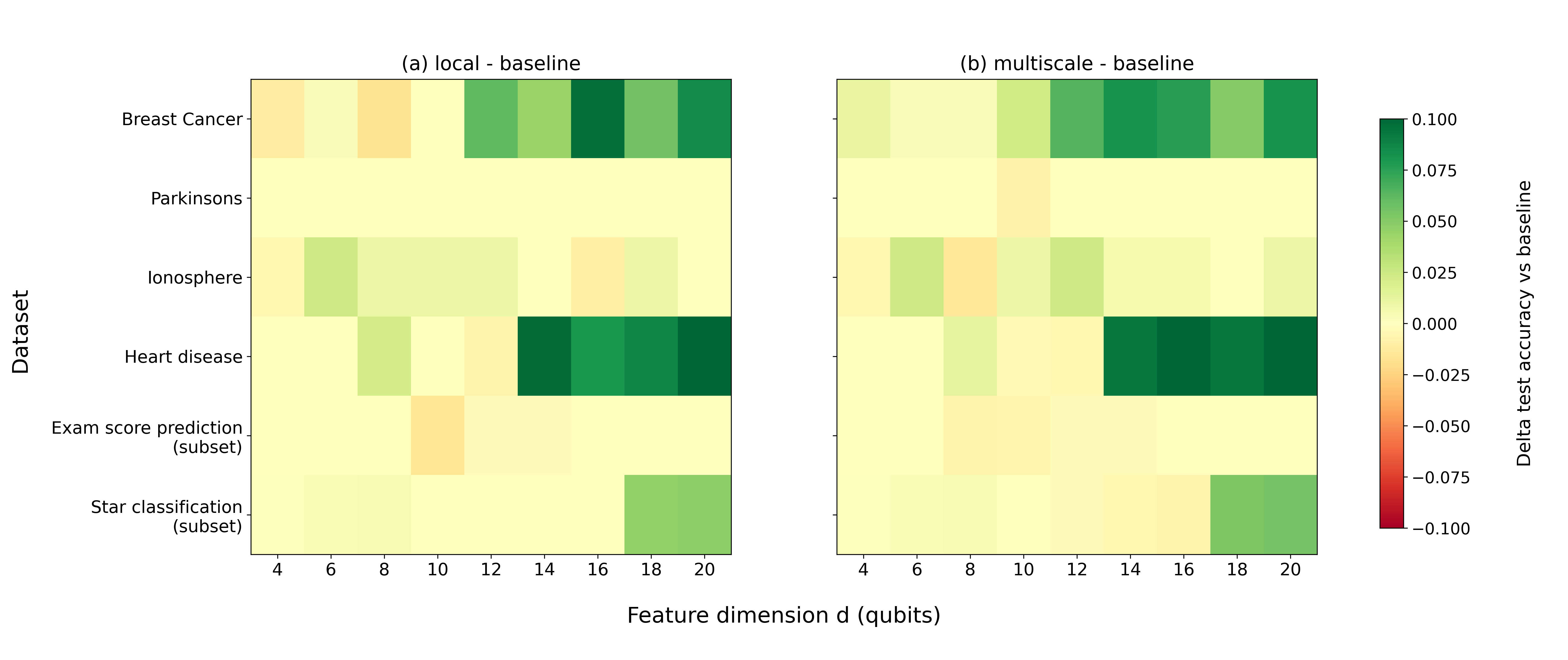}
\vspace{-5mm}
\caption{Heatmaps of test-accuracy differences relative to the
baseline kernel across datasets (rows) and feature dimensions $d$ (columns): (a) local minus baseline and (b) multi-scale minus baseline. Positive values indicate higher test accuracy than the baseline, whereas negative values indicate lower test accuracy.}
\label{fig:delta_testacc_heatmaps}
\end{figure*}

\paragraph{Per-dataset deltas across $d$.}
To connect the global heatmap view with per-dataset trends, Appendix~\ref{app:supplementary_figures} (Fig.~\ref{fig:delta_testacc_vs_d}) shows the test accuracy delta relative to baseline as a function of $d$.

\paragraph{Tradeoff plots: concentration vs.\ accuracy.}
Finally, Appendix~\ref{app:supplementary_figures} (Fig.~\ref{fig:tradeoff_scatter}) visualizes the empirical relationship between kernel concentration (off-diagonal p50) and test accuracy. Across datasets, these scatter plots illustrate that higher p50 (less concentration) can correlate with accuracy in some cases, but the relationship is not universal. Figure~\ref{fig:bar_mean_delta_test_acc} summarizes the mean test-accuracy delta relative to baseline for each dataset, averaged across $d\in\{4,6,\dots,20\}$.

\begin{figure}[t]
\centering
\includegraphics[width=0.9\linewidth]{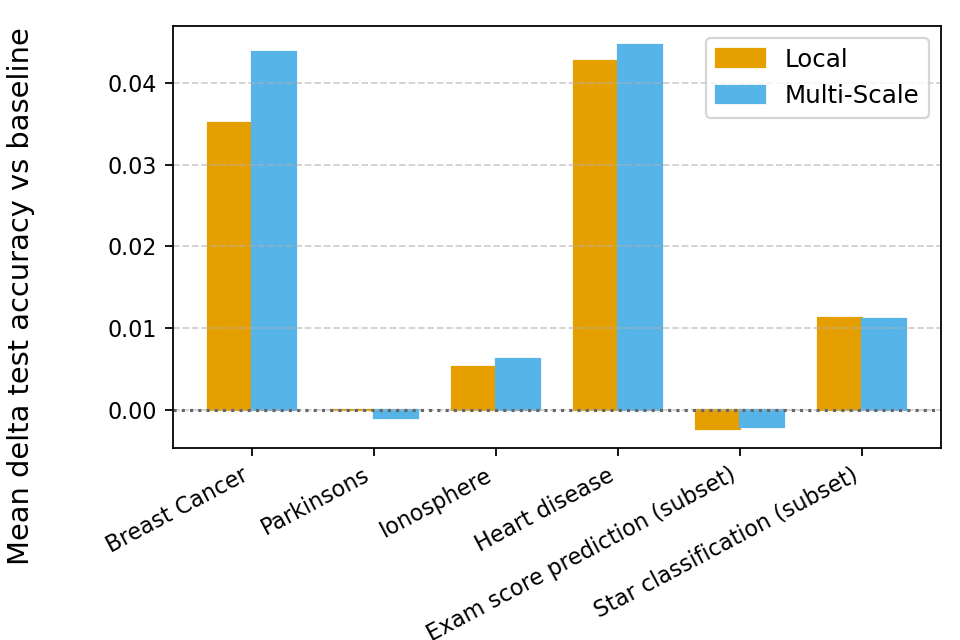}
\vspace{-3mm}
\caption{Mean test-accuracy difference relative to the baseline kernel for each dataset, averaged over feature dimensions $d\in\{4,6,\dots,20\}$. Positive values indicate that the local or multi-scale kernel improves accuracy on average, while negative values indicate a decrease relative to baseline.}
\label{fig:bar_mean_delta_test_acc}
\end{figure}

\section{Conclusion and outlook}\label{sec:conclusion}

We investigated two practical strategies designed to mitigate off-diagonal concentration in fidelity-based quantum kernels: local (patch-wise) constructions that aggregate subsystem similarities and multi-scale constructions that combine kernels across patch granularities. Within the Haar-random reference model, the mean raw global fidelity scales as $2^{-d}$, whereas the Hilbert--Schmidt similarity on a $k$-qubit patch has mean $2^{-k}$. Convex local aggregation preserves this patch-dependent mean scale, while multi-scale aggregation produces a weighted combination of the scales associated with its constituent patch sizes. These identities provide an analytical reference for the dimensional effect of locality. They concern mean unnormalized similarities, are not concentration bounds, and do not assume that the data-encoded states used in the experiments are Haar-random.

Across six tabular datasets and a common sweep in feature dimension $d\in\{4,6,\dots,20\}$, the local and multi-scale constructions consistently retained larger off-diagonal p50 and p95 values than the global fidelity baseline. The local kernel exhibited the strongest reduction in off-diagonal concentration, while the multi-scale kernel generally showed intermediate behavior. Effective rank followed a different pattern: the baseline consistently exhibited the highest effective rank, the local kernel generally exhibited a lower value, and the multi-scale kernel tended to lie between them. This distinction shows that effective rank measures spectral diversity and should not be interpreted as a direct proxy for off-diagonal concentration or kernel quality.

The downstream impact on SVM test accuracy was heterogeneous. Reduced concentration improved predictive performance for some datasets and dimensions but did not provide a dataset-independent guarantee. Preserving larger off-diagonal similarities is therefore not sufficient on its own; the retained variation must also be aligned with the label structure. Centered alignment, effective rank, and off-diagonal statistics provide complementary geometric diagnostics, but none of them individually certifies improved predictive performance.

Several directions remain open. Patch design and scale weighting could be selected in a data-driven manner, subject to constraints that preserve positive semidefiniteness. Extending these constructions to shot-based estimation and noisy hardware will require careful treatment of statistical error, noise mitigation, and the computational cost of reduced-state estimation. Locality and scale mixing therefore provide a practical basis for studying how quantum-kernel geometry affects learnability and resource tradeoffs within classical kernel pipelines.

\section{Acknowledgments}\label{acknowledgments}

This work was carried out as part of the Qiskit Advocate Mentorship Program (QAMP 2025). All authors participated as Qiskit Advocates on a voluntary basis. We thank the QAMP organizers for coordinating the program and providing a collaborative environment. This work received no external funding. Any opinions, findings, and conclusions expressed in this manuscript are those of the authors and do not necessarily reflect the views of IBM, Qiskit, or the Qiskit community.

\noindent\textit{Use of AI tools.} OpenAI's ChatGPT was used only for language polishing and to improve the clarity and readability of human-authored text. The authors reviewed and edited all AI-assisted output and take full responsibility for the accuracy, originality, and integrity of the manuscript.

\section{Code availability}

The source code, configuration files, and scripts used to reproduce the experimental results and figures reported in this work are publicly available in the accompanying repository \cite{quantum_kernels_qamp2026}.

\newpage

\bibliographystyle{unsrtnat}
\bibliography{refer}

\clearpage

\appendix

\onecolumngrid

\section{Local (patch-wise) kernel construction}\label{app:local_kernel_algo}

\begin{algorithm}[H]
\caption{Local Quantum Kernel Construction}
\label{alg:local_kernel}
\KwIn{
Dataset $X \in \mathbb{R}^{n \times d}$;
feature map circuit family $U_{\phi}(\cdot)$ with depth $L$;
patch set $\mathcal{P} = \{P_1,\dots,P_M\}$ (optional);
method $\in \{\textsc{Subcircuits}, \textsc{RDM}\}$;
RDM metric $\in \{\textsc{fidelity},\textsc{hs}\}$ (only for \textsc{RDM});
aggregation rule $\textsc{Agg} \in \{\textsc{mean},\textsc{weighted}\}$; weights $\{w_m\}$ required only for $\textsc{weighted}$ aggregation
}
\KwOut{Kernel matrix $K \in \mathbb{R}^{n \times n}$}

\If{$\mathcal{P}$ is not specified}{
Define default disjoint adjacent pairs $P_m=\{2m-2,2m-1\}$ for $m=1,\dots,d/2$
}

\If{$\textsc{Agg} = \textsc{weighted}$}{
Validate $w_m \ge 0$ and renormalize so that $\sum_m w_m = 1$
}

Initialize empty list $\mathcal{G}$\;

\uIf{\textsc{method = Subcircuits}}{
    \For{$m = 1$ \KwTo $M$}{
        Construct the patch feature map $U_{\phi}^{(P_m)}(\cdot)$ acting on $|P_m|$ qubits\;
        \For{$i = 1$ \KwTo $n$}{
            Form the patch feature subvector $\mathbf{x}_{i,P_m}$ by selecting components of $\mathbf{x}_i$ indexed by $P_m$\;
            Prepare the patch state $\ket{\psi_{P_m}(\mathbf{x}_i)} = U_{\phi}^{(P_m)}(\mathbf{x}_{i,P_m})\ket{0}^{\otimes |P_m|}$\;
        }
        Compute patch Gram matrix $K^{(P_m)}$ with entries
        $K^{(P_m)}_{ij} = \left|\langle \psi_{P_m}(\mathbf{x}_i) \mid \psi_{P_m}(\mathbf{x}_j) \rangle\right|^2$\;
        Append $K^{(P_m)}$ to $\mathcal{G}$\;
    }
}
\uElseIf{\textsc{method = RDM}}{
    Construct full feature map $U_{\phi}(\cdot)$ on $d$ qubits\;
    \For{$i = 1$ \KwTo $n$}{
        Prepare full state $\ket{\psi_i} = U_{\phi}(\mathbf{x}_i)\ket{0}^{\otimes d}$\;
    }
    \For{$m = 1$ \KwTo $M$}{
        Set $P_m^{c}\leftarrow\mathcal{Q}_d\setminus P_m$\;
        \For{$i = 1$ \KwTo $n$}{
            Compute $\rho_{P_m}(\mathbf{x}_i) = \operatorname{Tr}_{P_m^{c}}\!\left(\ket{\psi_i}\bra{\psi_i}\right)$\;
        }
        Compute patch Gram matrix $K^{(P_m)}$ with entries
        $
        K^{(P_m)}_{ij} =
        \begin{cases}
        F\!\left(\rho_{P_m}(\mathbf{x}_i),\rho_{P_m}(\mathbf{x}_j)\right), & \textsc{fidelity} \\
        \operatorname{Tr}\!\left(\rho_{P_m}(\mathbf{x}_i)\rho_{P_m}(\mathbf{x}_j)\right), & \textsc{hs}
        \end{cases}
        $\;
        Append $K^{(P_m)}$ to $\mathcal{G}$\;
    }
}

Aggregate patch kernels:
$K \leftarrow \textsc{Agg}\left(\{K^{(P_m)}\}_{m=1}^{M}\right)$\;

Symmetrize $K \leftarrow (K + K^\top)/2$\;
Normalize $K$ to unit diagonal: $K_{ij} \leftarrow \frac{K_{ij}}{\sqrt{K_{ii}K_{jj}}}$\;
Project $K$ onto the PSD cone via eigenvalue clipping\;
Renormalize $K$ to unit diagonal\;

\Return $K$\;
\end{algorithm}

\newpage

\section{Multi-scale kernel construction}\label{app:multiscale_kernel_algo}

\begin{algorithm}[H]
\caption{Multi-Scale Quantum Kernel Construction}
\label{alg:multiscale_kernel}
\KwIn{
Dataset $X \in \mathbb{R}^{n \times d}$;
feature map circuit family $U_{\phi}(\cdot)$ with depth $L$;
scales $\{\mathcal{P}^{(s)}\}_{s=1}^{S}$ (optional);
scale weights $\{\alpha_s\}_{s=1}^{S}$ (optional, nonnegative);
normalize flag (default: true)
}
\KwOut{Kernel matrix $K \in \mathbb{R}^{n \times n}$}

\If{scales are not specified}{
Define default scales: \\
\,\,\,(i) disjoint adjacent pairs $\mathcal{P}^{(1)}=\{(0,1),(2,3),\dots\}$ \\
\,\,\,(ii) full system $\mathcal{P}^{(2)}=\{\mathcal{Q}_d\}$\;
Set $S\leftarrow 2$\;
}

Validate scales (nonempty patches, indices in
$\{0,1,\dots,d-1\}$)\;

\If{scale weights are not specified}{
    Set $\alpha_s\leftarrow 1/S$ for all $s=1,\dots,S$\;
}
\Else{
    Validate $\alpha_s\geq 0$ for all $s$\;
    Validate $\sum_{s=1}^{S}\alpha_s>0$\;
    Renormalize $\alpha_s\leftarrow\alpha_s/\sum_{r=1}^{S}\alpha_r$ for all $s$\;
}

\For{$i = 1$ \KwTo $n$}{
Prepare full state $\ket{\psi_i} = U_{\phi}(\mathbf{x}_i)\ket{0}^{\otimes d}$\;
}

Initialize kernel $K \leftarrow 0_{n \times n}$\;

\ForEach{scale $s = 1,\dots,S$}{
    Initialize per-scale kernel $K^{(s)} \leftarrow 0_{n \times n}$\;
    \For{$m = 1$ \KwTo $M_s$}{
        \uIf{$|P_m^{(s)}| = d$}{
            \For{$i = 1$ \KwTo $n$}{
                \For{$j = i$ \KwTo $n$}{
                    Compute $v\leftarrow\left|\langle \psi_i \mid \psi_j \rangle\right|^2$\;
                    $K^{(s)}_{ij} \mathrel{+}= v$\;
                    \If{$j>i$}{
                        $K^{(s)}_{ji} \mathrel{+}= v$\;
                    }
                }
            }
        }
        \Else{
            Set $(P_m^{(s)})^{c}\leftarrow\mathcal{Q}_d\setminus P_m^{(s)}$\;
            \For{$i = 1$ \KwTo $n$}{
                Compute $\rho_{P_m^{(s)}}(\mathbf{x}_i) = \operatorname{Tr}_{(P_m^{(s)})^{c}}\!\left(\ket{\psi_i}\bra{\psi_i}\right)$\;
            }
            \For{$i = 1$ \KwTo $n$}{
                \For{$j = i$ \KwTo $n$}{
                    Compute $v\leftarrow\operatorname{Tr}\!\left(\rho_{P_m^{(s)}}(\mathbf{x}_i)\rho_{P_m^{(s)}}(\mathbf{x}_j)\right)$\;
                    $K^{(s)}_{ij} \mathrel{+}= v$\;
                    \If{$j>i$}{
                        $K^{(s)}_{ji} \mathrel{+}= v$\;
                    }
                }
            }
        }
    }
    Average over patches: $K^{(s)} \leftarrow \frac{1}{|\mathcal{P}^{(s)}|} K^{(s)}$\;
    Accumulate weighted contribution: $K \leftarrow K + \alpha_s K^{(s)}$\;
}

\If{normalize flag}{
    Validate $K_{ii}>0$ for all $i$\;
    Normalize $K$ to unit diagonal: $K_{ij}\leftarrow\frac{K_{ij}}{\sqrt{K_{ii}K_{jj}}}$\;
}
Symmetrize $K\leftarrow(K+K^\top)/2$\;

\Return $K$\;

\end{algorithm}

\newpage

\section{Additional analysis}\label{app:supplementary_figures}

\noindent The supplementary figures provide additional detail and alternative views of the main results in Section~\ref{results}. They expand on the concentration diagnostics, spectral structure, and performance trends across datasets and feature dimensions.

\subsection{Off-diagonal concentration tail (p95)} 

In Fig.~\ref{fig:offdiag_p95_vs_d}, we compare the off-diagonal concentration tail (p95) as the feature dimension (number of qubits) increases across six datasets. Lower p95 values indicate stronger suppression of off-diagonal correlations, corresponding to faster concentration. The Baseline method exhibits the most rapid decay, approaching zero by moderate dimensions, whereas the Local approach retains the largest off-diagonal values, indicating the slowest concentration. The Multi-Scale method consistently lies between these extremes, retaining larger off-diagonal values than the baseline but smaller values than the Local construction, with the same qualitative ordering across all datasets.

\begin{figure}[tbp]
\centering
\includegraphics[width=0.9\linewidth]{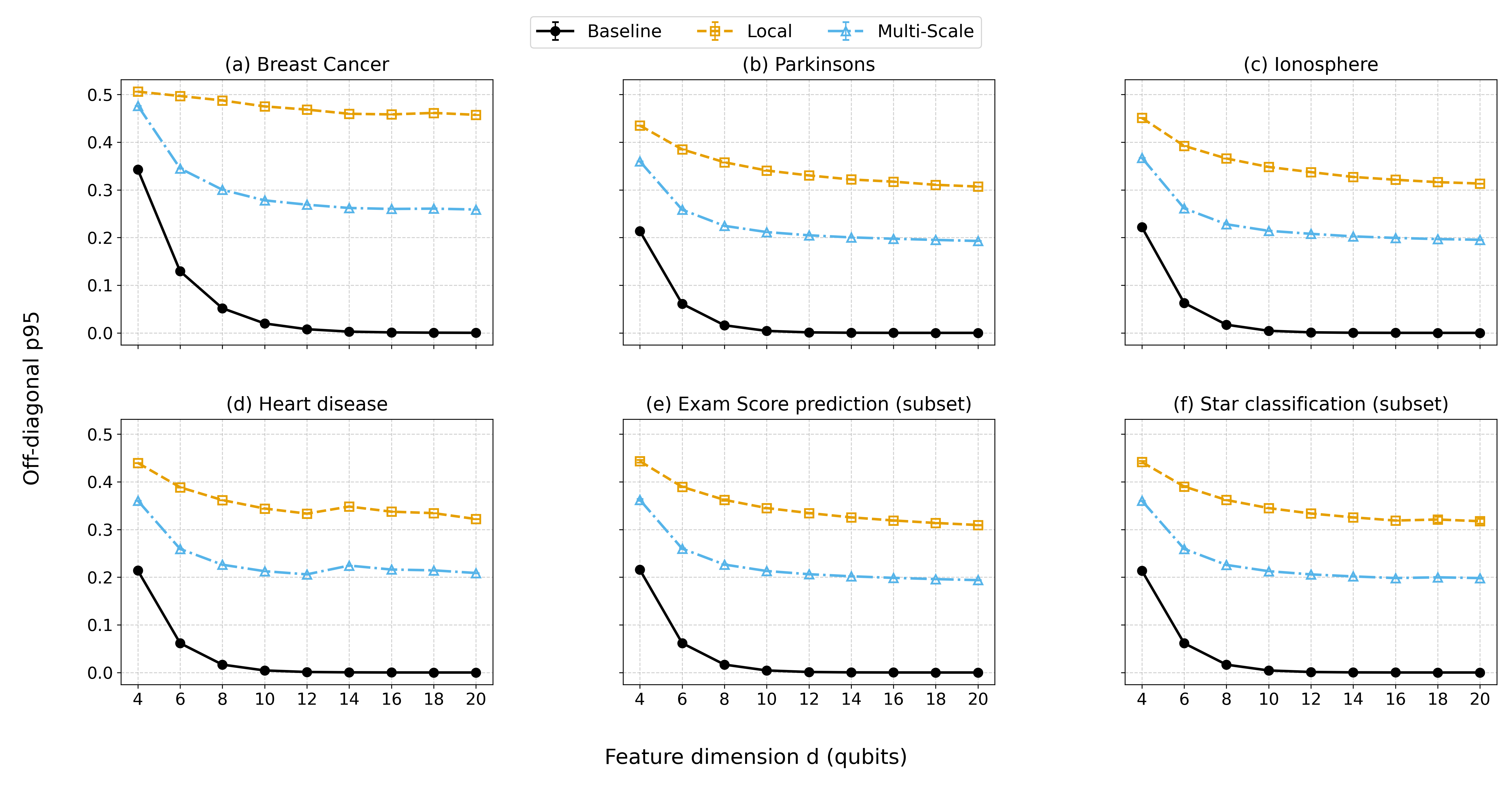}
\vspace{-5mm}
\caption{Off-diagonal concentration tail (p95) vs.\ feature dimension $d$. The p95 values of the global fidelity baseline decrease most rapidly toward zero, the local kernel retains the largest values, and the multi-scale kernel is generally intermediate. All Gram matrices are normalized to unit diagonal.}
\label{fig:offdiag_p95_vs_d}
\end{figure}

\subsection{Per-dataset test-accuracy deltas vs.\ \texorpdfstring{$d$}{d}}

In Fig.~\ref{fig:delta_testacc_vs_d}, we show the change in test accuracy of the Local and Multi-Scale kernels relative to the baseline as the feature dimension increases. Positive values indicate improved performance over the baseline, while negative values denote a slight degradation. The Multi-Scale kernel generally matches or exceeds the Local kernel, particularly at higher dimensions for the Breast Cancer, Heart Disease, and Star Classification datasets. In contrast, Parkinson’s, Ionosphere, and Exam Score Prediction exhibit minimal deviations, indicating that all methods perform comparably. Overall, the results show that the accuracy impact of the Multi-Scale kernel is heterogeneous: it improves performance in some dataset-dimension combinations, is comparable in others, and can decrease performance in some cases.

\begin{figure}[tbp]
\centering
\includegraphics[width=0.9\linewidth]{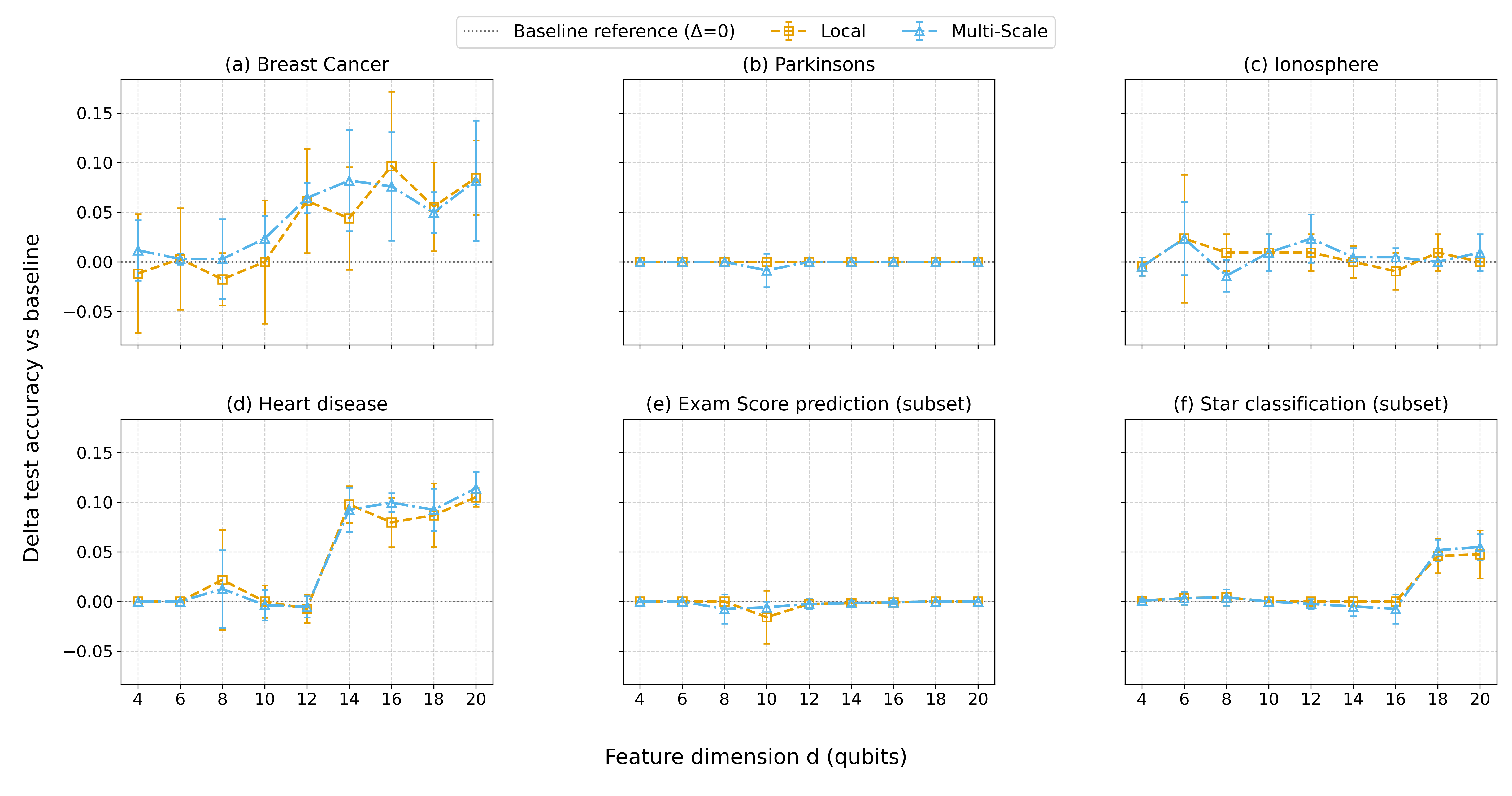}
\vspace{-5mm}
\caption{Test-accuracy difference relative to baseline vs.\ feature dimension $d$. Curves show the mean difference across the fixed seed grid, and error bars show the standard deviation. Positive values indicate higher test accuracy than the baseline, whereas negative values indicate lower test accuracy.}
\label{fig:delta_testacc_vs_d}
\end{figure}

\subsection{Tradeoff scatter: concentration (p50) vs.~test accuracy}

In Fig.~\ref{fig:tradeoff_scatter}, we illustrate the tradeoff between off-diagonal concentration (median, p50) and test accuracy for the Baseline, Local, and Multi-Scale kernels across six datasets. The Baseline exhibits the strongest concentration (lowest p50), while the Local kernel shows the weakest concentration (highest p50). The Multi-Scale kernel consistently occupies an intermediate regime in off-diagonal similarity magnitude. Although reduced concentration is associated with improved accuracy for some datasets, the relationship is not universal, indicating that reduced concentration alone does not guarantee better predictive performance.

\begin{figure}[tbp]
\centering
\includegraphics[width=0.9\linewidth]{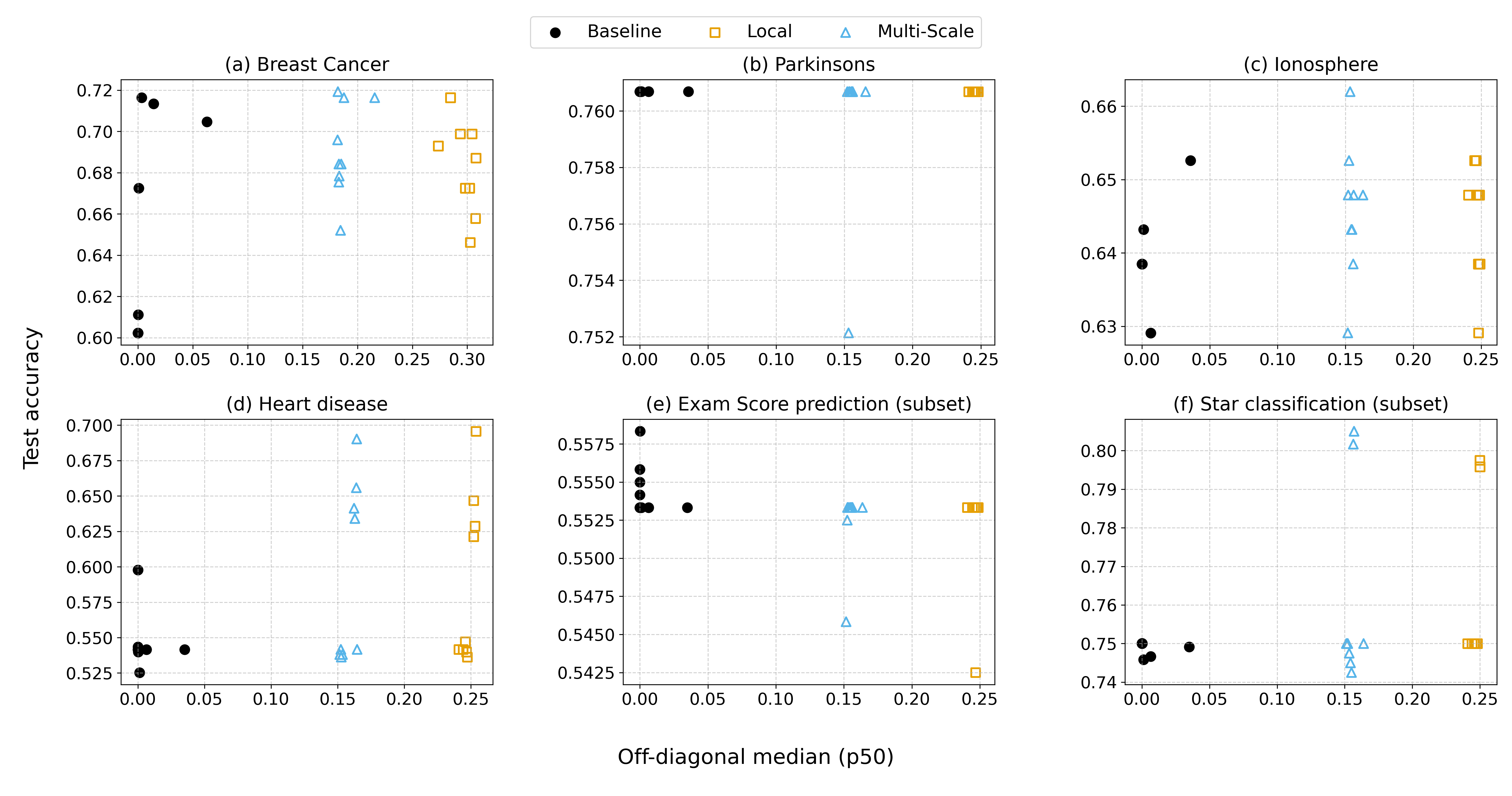}
\vspace{-5mm}
\caption{Relationship between median off-diagonal similarity (p50) and SVM test accuracy for the global fidelity baseline, local, and multi-scale kernels. Each panel corresponds to one dataset, and the points span the feature-dimension sweep. Higher p50 indicates less off-diagonal concentration, but its relationship with test accuracy is dataset- and dimension-dependent.}
\label{fig:tradeoff_scatter}
\end{figure}

\subsection{Representative eigen-spectra at \texorpdfstring{$d=12$}{d=12}}

In Fig.~\ref{fig:spectra_examples_d12}, we compare the normalized kernel eigenvalue spectra at feature dimension $d=12$ for the Baseline, Local, and Multi-Scale kernels across six datasets. Eigenvalues are normalized by their trace and plotted on a logarithmic scale to emphasize spectral shape. The Baseline typically exhibits a flatter spectrum, while the Local kernel displays a steeper decay, indicating stronger spectral concentration. The Multi-Scale kernel generally produces an intermediate spectral shape. These differences describe how spectral mass is allocated; they do not, by themselves, establish that one kernel is more informative for the prediction task.

\begin{figure}[tbp]
\centering
\includegraphics[width=0.9\linewidth]{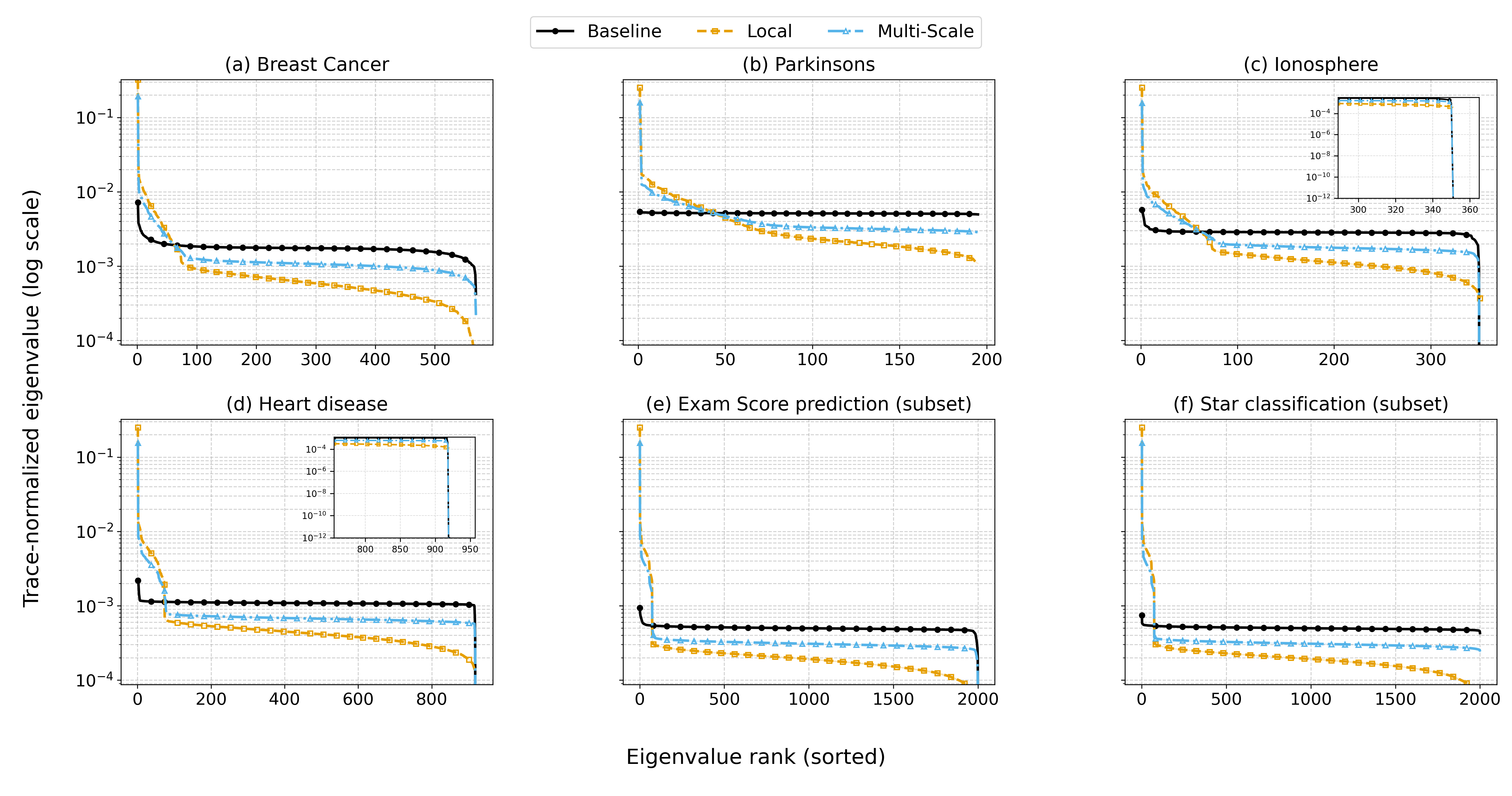}
\vspace{-5mm}
\caption{Representative trace-normalized eigenvalue spectra at $d=12$. Each panel compares the kernel eigen-spectrum for baseline, local, and multi-scale constructions on the same dataset and dimension. To compare spectral \emph{shape} rather than scale, for each run the eigenvalues are normalized by the trace ($\sum_i\lambda_i$) before aggregating across the fixed seed grid. The plotted curves show the mean across runs, and the vertical axis is shown on a semi-log scale.}
\label{fig:spectra_examples_d12}
\end{figure}

\subsection{Validation-selected \texorpdfstring{$C$}{C} vs.\ \texorpdfstring{$d$}{d}}\label{app:bestC_vs_d}

Fig.~\ref{fig:bestC_vs_d} reports the value of the SVM regularization parameter $C$ selected by validation for each dataset and feature dimension $d$, for each kernel family. This diagnostic is included to assess the stability of model selection across dimensions and kernels. Across most datasets, the selected $C$ values remain broadly consistent across kernel families, indicating that the observed performance differences primarily reflect the underlying kernel representations rather than changes in classifier regularization.

\begin{figure}[tbp]
\centering
\includegraphics[width=0.9\linewidth]{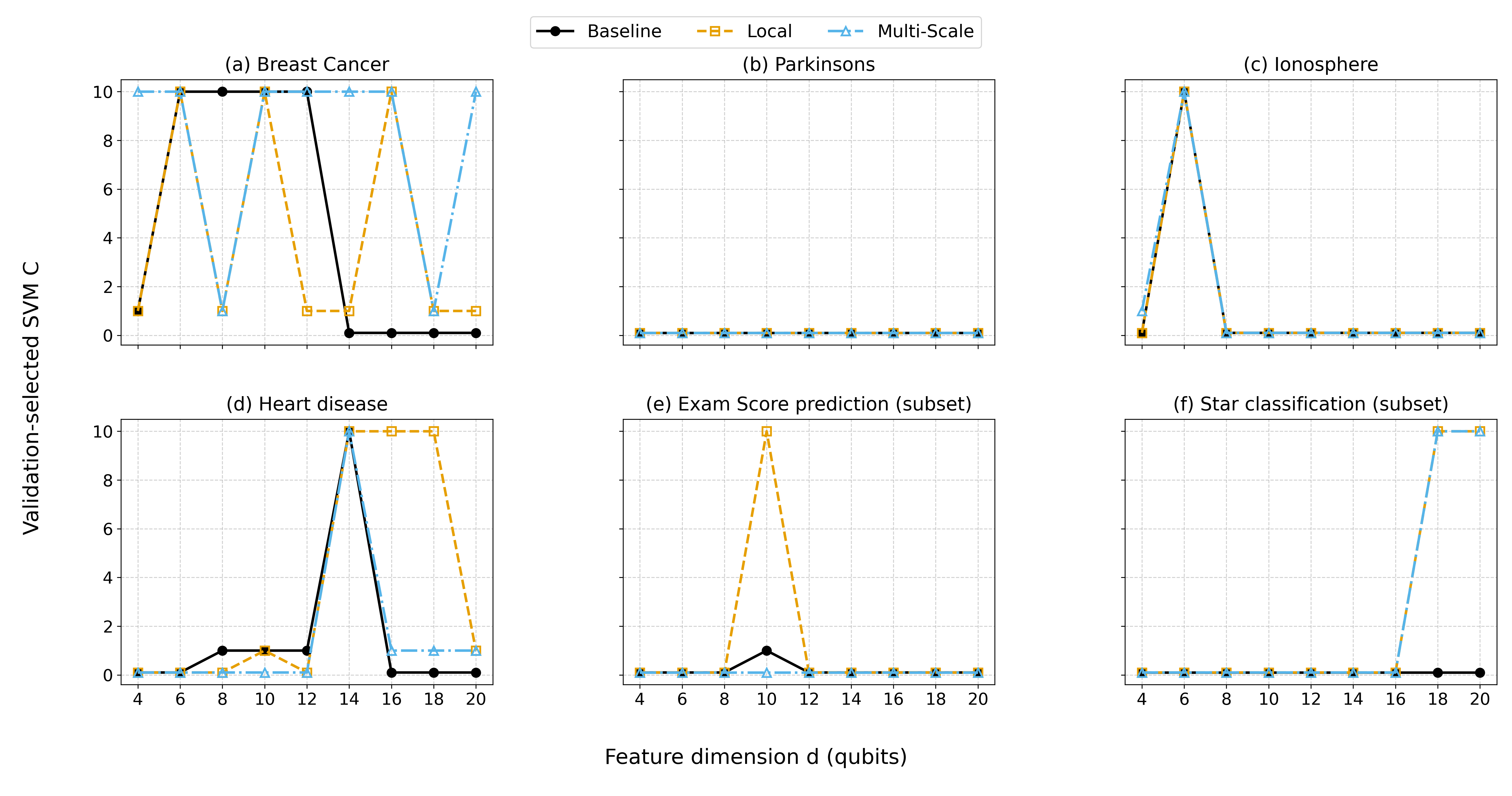}
\vspace{-5mm}
\caption{Validation-selected SVM regularization parameter $C$ vs.\ feature dimension $d$. For each dataset, dimension, and kernel family, each point shows the mode across the fixed seed grid of the $C$ values selected independently by validation for each seed from the grid in Eq.~\eqref{eq:svm_C_grid}.}
\label{fig:bestC_vs_d}
\end{figure}

\section{Proofs of the mean-scale theorems}\label{app:proofs}

\subsection{Proof of Theorem~\ref{thm:patch-concen}: Patch-wise mean scale}

\begin{proof}
For a Haar-random pure state, the Haar first-moment identity \cite{mele2024introduction} gives
\begin{equation}
    \mathbb{E}_{\varphi}\!\left[
    \ket{\varphi}\bra{\varphi}
    \right]
    =
    \frac{I_{2^d}}{2^d}.
\end{equation}
Taking the partial trace over $P_m^c$ yields
\begin{equation}
    \mathbb{E}_{\varphi}\!\left[
    \rho_{P_m}^{\varphi}
    \right]
    =
    \operatorname{Tr}_{P_m^c}\!\left(
    \frac{I_{2^d}}{2^d}
    \right)
    =
    \frac{I_{2^k}}{2^k}.
\end{equation}
Using the independence of $\ket{\psi}$ and $\ket{\varphi}$, and writing $\mathbb{E}_{\psi}$ and $\mathbb{E}_{\varphi}$ for the expectations over the corresponding Haar-random states, we obtain
\begin{align}
\mathbb{E}_{\psi,\varphi}\!\left[
\kappa_{P_m}(\psi,\varphi)
\right]
&=
\mathbb{E}_{\psi}\!\left[
\operatorname{Tr}\!\left(
\rho_{P_m}^{\psi}
\,
\mathbb{E}_{\varphi}\!\left[
\rho_{P_m}^{\varphi}
\right]
\right)
\right]
\\
&=
\mathbb{E}_{\psi}\!\left[
\operatorname{Tr}\!\left(
\rho_{P_m}^{\psi}
\frac{I_{2^k}}{2^k}
\right)
\right]
\\
&=
\frac{1}{2^k}
\mathbb{E}_{\psi}\!\left[
\operatorname{Tr}\!\left(
\rho_{P_m}^{\psi}
\right)
\right]
\\
&=
2^{-k},
\end{align}
where $\operatorname{Tr}(\rho_{P_m}^{\psi})=1$.
\end{proof}

\subsection{Proof of Theorem~\ref{thm:local_concen}: Local kernel mean scale}

\begin{proof}
By the definition of the local kernel in Eq.~\eqref{eq:local_kernel} and linearity of expectation,
\begin{align}
\mathbb{E}_{\psi,\varphi}\!\left[
k_{\mathrm{loc}}(\psi,\varphi)
\right]
&=
\mathbb{E}_{\psi,\varphi}\!\left[
\sum_{m=1}^{M}
w_m\,
\kappa_{P_m}(\psi,\varphi)
\right]
\\
&=
\sum_{m=1}^{M}
w_m\,
\mathbb{E}_{\psi,\varphi}\!\left[
\kappa_{P_m}(\psi,\varphi)
\right].
\end{align}
Since every patch has size $k$, Theorem~\ref{thm:patch-concen} gives
\begin{equation}
\mathbb{E}\!\left[
\kappa_{P_m}(\psi,\varphi)
\right]
=
2^{-k}
\end{equation}
for every $m$. Therefore,
\begin{align}
\mathbb{E}\!\left[
k_{\mathrm{loc}}(\psi,\varphi)
\right]
&=
2^{-k}
\sum_{m=1}^{M}w_m
\\
&=
2^{-k},
\end{align}
where the last equality follows from
$\sum_{m=1}^{M}w_m=1$.
\end{proof}

\subsection{Proof of Theorem~\ref{thm:multiscale_concen}: Multi-scale kernel mean scale}

\begin{proof}
By the multi-scale construction in Eq.~\eqref{eq:multiscale_kernel} and linearity of expectation,
\begin{align}
\mathbb{E}_{\psi,\varphi}\!\left[
k_{\mathrm{ms}}(\psi,\varphi)
\right]
&=
\mathbb{E}_{\psi,\varphi}\!\left[
\sum_{s=1}^{S}
\alpha_s\,k^{(s)}(\psi,\varphi)
\right]
\\
&=
\sum_{s=1}^{S}
\alpha_s\,
\mathbb{E}_{\psi,\varphi}\!\left[
k^{(s)}(\psi,\varphi)
\right].
\end{align}
Applying Lemma~\ref{lem:scale_concen} gives
\begin{equation}
\mathbb{E}_{\psi,\varphi}\!\left[
k_{\mathrm{ms}}(\psi,\varphi)
\right]
=
\sum_{s=1}^{S}
\alpha_s\,2^{-k_s},
\end{equation}
which proves the result.
\end{proof}

\twocolumngrid

\end{document}